\documentclass[journal,11pt,a4paper,onecolumn]{IEEEtran}
\usepackage{graphicx}
\usepackage{amssymb,amsbsy}
\usepackage{epstopdf}
\usepackage{amsmath,amsthm}
\usepackage{enumerate}
\usepackage{eufrak}
\usepackage{cite}
\usepackage{mathcomp}
\usepackage{supertabular}
\usepackage{longtable}
\usepackage{stmaryrd}
\usepackage{url}
\usepackage{color}
\usepackage{rotating}
\usepackage{float}

\usepackage{array}
\usepackage{mathtools}
\usepackage{multicol}
\usepackage{amsmath}
\usepackage{algorithm}
\usepackage{algorithmic}
\usepackage{amssymb}
\usepackage{xcolor}

\usepackage[all]{xy}
\entrymodifiers={++[o][F-]}

\usepackage{setspace}
\theoremstyle{definition}

\newtheorem{prop}{Proposition}[section]
\newtheorem{thm}[prop]{Theorem}
\newtheorem{cor}[prop]{Corollary}
\newtheorem{lem}[prop]{Lemma}

\newtheorem{conj}[prop]{Conjecture}

\newtheorem{exa}{Example}[section]

\newtheorem{rem}{Remark}

\DeclareMathOperator{\lcm}{lcm}

\begin{document}

\newcommand{\vA}{{\bf A}}
\newcommand{\vAtilde}{\widetilde{\bf A}}

\newcommand{\vB}{{\bf B}}
\newcommand{\vBtilde}{\widetilde{\bf B}}

\newcommand{\vC}{{\bf C}}
\newcommand{\vD}{{\bf D}}
\newcommand{\vH}{{\bf H}}
\newcommand{\vI}{{\bf I}}

\newcommand{\vY}{{\bf Y}}
\newcommand{\vZ}{{\bf Z}}

\newcommand{\vJ}{{\bf J}}

\newcommand{\vM}{{\bf M}}
\newcommand{\vN}{{\bf N}}
\newcommand{\vU}{{\bf U}}
\newcommand{\vV}{{\bf V}}
\newcommand{\vT}{{\bf T}}
\newcommand{\vR}{{\bf R}}
\newcommand{\vS}{{\bf S}}

\newcommand{\va}{{\bf a}}
\newcommand{\vb}{{\bf b}}
\newcommand{\vc}{{\bf c}}

\newcommand{\ve}{{\bf e}}
\newcommand{\vh}{{\bf h}}
\newcommand{\vp}{{\bf p}}

\newcommand{\vu}{{\bf u}}
\newcommand{\vv}{{\bf v}}
\newcommand{\vw}{{\bf w}}
\newcommand{\vx}{{\bf x}}
\newcommand{\vhx}{{\widehat{\bf x}}}
\newcommand{\vtx}{{\widetilde{\bf x}}}
\newcommand{\vy}{{\bf y}}
\newcommand{\vz}{{\bf z}}

\newcommand{\vj}{{\bf j}}
\newcommand{\vzero}{{\bf 0}}
\newcommand{\vone}{{\bf 1}}
\newcommand{\vbeta}{{\boldsymbol \beta}}
\newcommand{\vchi}{{\boldsymbol \chi}}

\newcommand{\tA}{\textrm A}
\newcommand{\tB}{\textrm B}
\newcommand{\A}{\mathcal A}
\newcommand{\B}{\mathcal B}
\newcommand{\C}{\mathcal C}
\newcommand{\D}{\mathcal D}
\newcommand{\E}{\mathcal E}
\newcommand{\F}{\mathcal F}
\newcommand{\G}{\mathcal G}
\newcommand{\M}{\mathcal M}
\newcommand{\HH}{\mathcal H}
\newcommand{\PP}{\mathcal P}

\newcommand{\Q}{\mathcal Q}
\newcommand{\Qb}{\bar{\mathcal Q}}
\newcommand{\Db}{{\bar{\Delta}}}

\newcommand{\pQ}{{\bf p}\mathcal Q}
\newcommand{\pQb}{{\bf p}\bar{\mathcal Q}}

\newcommand{\R}{\mathcal R}
\newcommand{\SSS}{\mathcal S}
\newcommand{\U}{\mathcal U}
\newcommand{\V}{\mathcal V}
\newcommand{\Y}{\mathcal Y}
\newcommand{\Z}{\mathcal Z}

\newcommand{\Pg}{{{\mathcal P}_{\rm gram}}}
\newcommand{\Pgint}{{{\mathcal P}^\circ_{\rm gram}}}
\newcommand{\Pgrc}{{{\mathcal P}_{\rm GRC}}}
\newcommand{\Pgrcint}{{{\mathcal P}^\circ_{\rm GRC}}}
\newcommand{\Pint}{{{\mathcal P}^\circ}}
\newcommand{\Ag}{{\bf A}_{\rm gram}}

\newcommand{\CC}{\mathbb C} 
\newcommand{\RR}{\mathbb R}
\newcommand{\ZZ}{\mathbb Z}
\newcommand{\FF}{\mathbb F}
\newcommand{\KK}{\mathbb K}

\newcommand{\Fnd}{\FF_q^{n^{\otimes d}}}
\newcommand{\Knd}{\KK^{n^{\otimes d}}}

\newcommand{\ceiling}[1]{\left\lceil{#1}\right\rceil}
\newcommand{\floor}[1]{\left\lfloor{#1}\right\rfloor}
\newcommand{\bbracket}[1]{\left\llbracket{#1}\right\rrbracket}

\newcommand{\inprod}[1]{\left\langle{#1}\right \rangle}


\newcommand{\beas}{\begin{eqnarray*}} 
\newcommand{\eeas}{\end{eqnarray*}} 

\newcommand{\bm}[1]{{\mbox{\boldmath $#1$}}} 

\newcommand{\sizeof}[1]{\left\lvert{#1}\right\rvert}
\newcommand{\wt}{{\rm wt}} 
\newcommand{\supp}{{\rm supp}} 
\newcommand{\dg}{d_{\rm gram}} 
\newcommand{\da}{d_{\rm asym}} 
\newcommand{\dist}{{\rm dist}} 
\newcommand{\ssyn}{s_{\rm syn}}
\newcommand{\sseq}{s_{\rm seq}}
\newcommand{\nullplus}{{\rm Null}_{>\vzero}}

\newcommand{\tworow}[2]{\genfrac{}{}{0pt}{}{#1}{#2}}
\newcommand{\qbinom}[2]{\left[ {#1}\atop{#2}\right]_q}

\newcommand{\Lovasz}{Lov\'{a}sz }
\newcommand{\etal}{\emph{et al.}}

\newcommand{\todo}{{\color{red} (TODO) }}


\title{Codes for DNA Sequence Profiles}

\author{Han Mao Kiah, Gregory J.~Puleo, and Olgica Milenkovic\\
\thanks{This work was supported in part by the NSF STC Class 2010 CCF 0939370 grant and the Strategic Research Initiative (SRI) Grant conferred by the University of Illinois, Urbana-Champaign. Research of the second author was supported by the IC Postdoctoral Research Fellowship. This work has been submitted in part to ISIT 2015 and will appear in part at ITW 2015.}
Department of Electrical and Computer Engineering, University of Illinois, Urbana-Champaign
}

\maketitle

\begin{abstract}
We consider the problem of storing and retrieving information from synthetic DNA media.
The mathematical basis of the problem is the construction and design of sequences that may be discriminated 
based on their collection of substrings observed through a noisy channel.
This problem of reconstructing sequences from traces was first investigated in the noiseless setting under the name of ``Markov type'' analysis. Here, we explain the connection between the reconstruction problem and the problem of DNA synthesis and sequencing, and introduce the notion of a DNA storage channel. We analyze the number of sequence equivalence classes under the channel mapping and propose new asymmetric coding techniques to combat the effects of synthesis and sequencing noise. In our analysis, we make use of restricted de Bruijn graphs and Ehrhart theory for rational polytopes.
\end{abstract}

%


\section{Introduction}

Reconstructing sequences based on partial information about their subsequences, substrings, or composition is an important problem arising in channel synchronization systems, phylogenomics, genomics, and proteomic sequencing~\cite{kannan2005more,acharya2010reconstructing,acharya2014quadratic}. With the recent development of archival DNA-based storage devices~\cite{Church.etal:2012,Goldman.etal:2013} and rewritable, random-access storage media~\cite{Ma.etal:pending}, a new family of reconstruction questions has emerged regarding how to \emph{design sequences} which can be easily and accurately reconstructed based on their substrings, in the presence of read and write errors. 
The write process reduces to DNA synthesis, while the read process involves both DNA sequencing and assembly. The assembly procedure is NP-hard under most formulations~\cite{Medvedev.etal:2007}. Nevertheless, practical approximation algorithms based on Eulerian paths in de Bruijn graphs have shown to offer good reconstruction performance under the high-coverage model~\cite{Compeau.etal:2011}.

In the setting we propose to analyze, one first synthesizes a sequence $\vx \in \mathcal{D}=\{{A,T,G,C\}}^n$, and then fragments it in the process of sequencing into a collection of substrings of approximately the same length, $\ell$. These substrings are often referred to as \emph{reads}. In practice, the length $\ell$ ranges anywhere between $100$ to $1500$ nts\footnote{For our system currently under development, due to the high cost of synthesis, we have chosen $n=1000$. In addition, we used multiple sequences to increase storage capacity.}.
Ideally, one would like to synthesize $\vx$ and sequence all $\ell$-substrings without errors, which is not possible in practice. For large $n$, the synthesis error-rate of $\vx$ is roughly $1-3\%$. Substrings of short length may be sequenced with an error-rate not exceeding $1\%$; long substrings exhibit much higher sequencing error-rates, often as high as $15\%$. Furthermore, due to non-uniform fragmentation, a number of the substrings are not available for sequencing, leaving coverage gaps in the original message. 

To model this read-write phenomenon, we introduce the notion of a \emph{DNA storage channel} that takes as its input a sequence $\vx$ of length $n$, introduces $\ssyn$ substitution errors in $\vx$, with the resulting sequence denoted by $\vtx$. The channel proceeds to output all or a subset of substrings of the sequence $\vtx$ of length $\ell$, $\ell<n$. Each of the substrings is allowed to have additional substitution errors, due to sequencing. The total number of substring sequencing errors equals $\sseq$. 
The substrings at the output of the DNA storage channel are collectively enumerated by a vector $\vhx$, termed the channel output (see Fig.~\ref{fig:DNAstorage} for an illustration).

The main contributions of the paper are as follows. The first contribution is to {\em model the read process (sequencing)} through the use of {\em profile vectors}.
A profile vector of a sequence enumerates all substrings of the sequence, and profiles form a pseudometric space amenable for coding theoretic analysis. The second contribution of the paper is to \emph{introduce a new family of codes} for the three classes of errors arising in the DNA storage channel due to synthesis, lack of coverage and sequencing, and show that they may be characterized by {\em asymmetric errors} studied in classical coding theory. Our third contribution is a code design technique which makes use of
(a) codewords with different profile vectors or profile vectors at sufficiently large distance from each other; and 
(b) codewords with $\ell$-substrings of high biochemical stability which are also resilient to errors. For this purpose, we consider a number of {\em codeword constraints} known to influence the performance of both the synthesis and sequencing systems, one of which we termed the \emph{balanced content constraint}.

For the case when we allow arbitrary $\ell$-substrings, the problem of enumerating all valid profile vectors 
was previously addressed by Jacquet \etal{} \cite{Jacquet.etal:2012} in the context of ``Markov types''. 
However, the method of Jacquet \etal{} does not extend to the case of enumeration of profiles with specific $\ell$-substring constraints or profiles at sufficiently large distance from each other.
We cast our more general enumeration and code design question as 
a problem of {\em enumerating integer points in a rational polytope}
and use tools from {\em Ehrhart theory} to provide estimates of the sizes of the underlying codes.
We also describe two decoding procedures for sequence profiles that combine graph theoretical principles
and sequencing by hybridization methods.

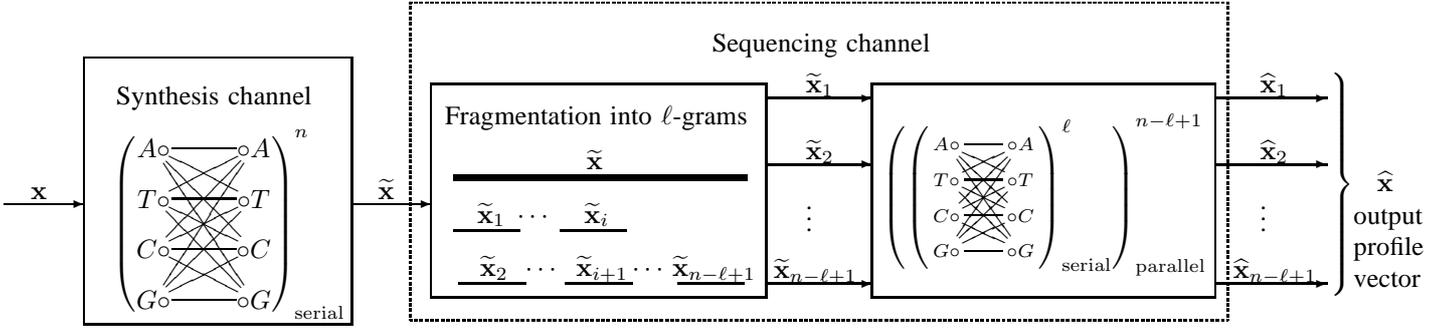
\begin{figure*}[t]
\small
\begin{center}
\begin{picture}(480,100)(0,0)

\linethickness{0.5pt}
\put(5,-15){\framebox(100,100){}}
\put(17,70){\txt{Synthesis channel}}
\put(12,20){\mbox{
$\left(
\vcenter{\xymatrix @R=2mm{
*=0++[o]{A\circ} \ar@{-}[r]\ar@{-}[dr]\ar@{-}[ddr]\ar@{-}[dddr]&*=0++[o]{\circ A}\\
*=0++[o]{T\circ} \ar@{-}[r]\ar@{-}[ur]\ar@{-}[dr]\ar@{-}[ddr]&*=0++[o]{\circ T}\\
*=0++[o]{C\circ} \ar@{-}[r]\ar@{-}[ur]\ar@{-}[uur]\ar@{-}[dr]&*=0++[o]{\circ C}\\
*=0++[o]{G\circ} \ar@{-}[r]\ar@{-}[ur]\ar@{-}[uur]\ar@{-}[uuur]&*=0++[o]{\circ G}}}
\right)^n_{\rm serial}$
}}
\put(135,-5){\framebox(125,80){}}
\put(140,60){\mbox{Fragmentation into $\ell$-grams}}
\linethickness{2pt}
\put(143,40){\line(1,0){110}}
\put(193,43){\mbox{$\vtx$}}
\linethickness{0.5pt}
\put(143,20){\line(1,0){25}}
\put(183,20){\line(1,0){25}}
\put(145,0){\line(1,0){25}}
\put(185,0){\line(1,0){25}}
\put(227,0){\line(1,0){25}}
\put(152,23){\mbox{$\vtx_1$}}
\put(153,3){\mbox{$\vtx_2$}}
\put(192,23){\mbox{$\vtx_{i}$}}
\put(189,3){\mbox{$\vtx_{i+1}$}}
\put(225,3){\mbox{$\vtx_{n-\ell+1}$}}
\put(210,3){\mbox{$\cdots$}}
\put(170,3){\mbox{$\cdots$}}
\put(167,23){\mbox{$\cdots$}}
\put(240,90){\txt{Sequencing channel}}
\put(300,-5){\framebox(128,80){}}
\put(127,-14){\dashbox(306,120){}}
\put(300,30){\mbox{
$\left(\left(
\def\objectstyle{\scriptstyle}
\vcenter{\xymatrix @R=0mm@C=5mm{
*=0++[o]{A\circ} \ar@{-}[r]\ar@{-}[dr]\ar@{-}[ddr]\ar@{-}[dddr]&*=0++[o]{\circ A}\\
*=0++[o]{T\circ} \ar@{-}[r]\ar@{-}[ur]\ar@{-}[dr]\ar@{-}[ddr]&*=0++[o]{\circ T}\\
*=0++[o]{C\circ} \ar@{-}[r]\ar@{-}[ur]\ar@{-}[uur]\ar@{-}[dr]&*=0++[o]{\circ C}\\
*=0++[o]{G\circ} \ar@{-}[r]\ar@{-}[ur]\ar@{-}[uur]\ar@{-}[uuur]&*=0++[o]{\circ G}}}
\right)^{\ell}_{\rm serial}\right)^{n-\ell+1}_{\rm parallel}$
}}
\put(-25,30){\vector(1,0){30}}
\put(-15,32){\mbox{$\vx$}}
\put(105,30){\vector(1,0){30}}
\put(115,32){\mbox{$\vtx$}}

\put(260,0){\vector(1,0){40}}
\put(260,45){\vector(1,0){40}}
\put(260,70){\vector(1,0){40}}

\put(263,2){\mbox{$\vtx_{n-\ell+1}$}}
\put(275,20){\mbox{$\vdots$}}
\put(275,47){\mbox{$\vtx_2$}}
\put(275,72){\mbox{$\vtx_1$}}

\put(428,0){\vector(1,0){42}}
\put(428,45){\vector(1,0){42}}
\put(428,70){\vector(1,0){42}}

\put(435,2){\mbox{$\vhx_{n-\ell+1}$}}
\put(445,20){\mbox{$\vdots$}}
\put(445,47){\mbox{$\vhx_2$}}
\put(445,72){\mbox{$\vhx_1$}}
\put(470,35){\mbox{$\left.\vphantom{\begin{array}{c}a\\a\\a\\a\\a\\a\\a\end{array}}\right\}~~\vhx$}}
\put(480,11){\parbox[c][4em][c]{30px}{output profile vector}}

\end{picture}
\end{center}
\caption{The DNA Storage Channel. Information is encoded in a DNA sequence $\vx$ which is synthesized with potential errors. The output of the synthesis process is $\vtx$. During readout, the sequence $\vtx$ is read through the sequencing channel, which fragments the sequence and possibly perturbs the fragments via substitution error. The output of the channel is a set of DNA fragments, along with their frequency count.}
\label{fig:DNAstorage}
\end{figure*}

\section{Profile Vectors and a Metric Space}\label{sec:profile}

%

Let $\llbracket q\rrbracket$ denote the set of integers $\{0,1,2,\ldots, q-1\}$ and 
consider a word $\vx$ of length $n$ over $\llbracket q\rrbracket$. 
Suppose that $\ell<n$. An {\em $\ell$-gram} is a substring of $\vx$ of length $\ell$.
Let $\vp(\vx;q,\ell)$ denote the ($\ell$-gram) {\em profile vector} of length $q^\ell$, indexed by all words of $\llbracket q\rrbracket^{\ell}$ ordered lexicographically. In the profile vector, an entry indexed by $\vz$ gives the number of occurrences of $\vz$ as an $\ell$-gram of $\vx$.
For example, $\vp(0000;2,2)=(3,0,0,0)$, while $\vp(0101;2,2)=(0,2,1,0)$. Observe that for any $\vx\in \llbracket q\rrbracket^\ell$, 
the sum of entries in $\vp(\vx;q,\ell)$ equals $(n-\ell+1)$.

Suppose that the data of interest is encoded by a vector  $\vx\in \llbracket q\rrbracket^n$ and 
let $\vhx$ be the output profile of the DNA channel. In what follows, we characterize the error vector $\ve\triangleq\vp(\vx;q,\ell)-\vhx$ which arises in the process of DNA-based data storage and the type of errors captured by this vector. 
\begin{enumerate}[(i)]
\item {\bf Substitution errors due to synthesis}. Here, certain symbols in the word $\vx$ may be changed as a result of erroneous synthesis. If one symbol is changed, in the perfect coverage case, $\ell$ $\ell$-grams will decrease their counts by one and $\ell$ $\ell$-grams will increase their counts by one. 
Hence, the error resulting from $\ssyn$ substitutions equals $\ve=\ve_--\ve_+$, 
where $\ve_+,\ve_-\ge \vzero$, and both vectors have weight $\ssyn \, \ell$.

\item {\bf Substitution errors due to sequencing}. Here, certain symbols in each fragment $\vtx_i$ may be changed during the sequencing process. Suppose the $\ell$-gram $\vtx_i$ is altered to $\vhx_i$ , $\vhx_i\ne\vtx_i$. Then the count for $\vtx_i$ will decrease by one while the count for $\vhx_i$ will increase by one.
Hence, the error resulting from $\sseq$ substitutions equals $\ve=\ve_--\ve_+$, 
where $\ve_+,\ve_-\ge \vzero$, and both vectors have weight $\sseq$.
\item {\bf Undersampling errors}. Such errors occur when not all $\ell$-grams are observed during fragmentation and subsequently sequenced.
For example, suppose that $\vx=00000$, and that $\vhx$ is the channel output $3$-gram profile vector.
Undersampling of one $3$-gram results in the weight of $\vhx_{000}$ being four instead of five.
Note that undersampling of $t$ $\ell$-grams results in 
an asymmetric error $\ve \geq \vzero$ of weight $t$.
\end{enumerate}

Consider further a subset $S\subseteq \bbracket{q}^\ell$.
For $\vx\in\bbracket{q}^n$, we similarly define $\vp(\vx;S)$ to be the vector indexed by $S$,
whose entry indexed by $\vz$ gives the number of occurrences of $\vz$ as an $\ell$-gram of $\vx$.
We are interested in vectors $\vx$ whose $\ell$-grams belong to $S$.
Once again, the sum of entries in $\vp(\vx;S)$ equals $n-\ell+1$.

The choice of $S$ is governed by certain considerations in DNA sequence design, including 
\begin{enumerate}[(i)]
\item {\bf Weight profiles of $\ell$-grams}. For the application at hand, one may want to choose $S$ to consist of $\ell$-grams with a fixed proportion of $C$ and $G$ bases, as 
this proportion -- known as the GC-content of the sequence -- influences the thermostability, folding processes and
overall coverage of the $\ell$-grams. From the perspective of sequencing, GC contents of roughly $50\%$ are desired.

To make this modeling assumption more precise and general, 
we assume sets $S$ of the form described below. Suppose that $0\le w_1< w_2\le\ell$ 
and $1\le q^*\le q-1$. Let $[w_1,w_2]$ denote the set of integers $\{w_1,w_1+1,\ldots, w_2\}$. 
For each $\vx\in\bbracket{q}^\ell$, 
let the {\em $q^*$-weight} of $\vx$ be the number of symbols in $\vx$ that belong to $[q-q^*,q-1]$, 
and denote the weight by $\wt(\vx;q^*)$. 
Let
\[
S(q,\ell; q^*,[w_1,w_2])\triangleq \left\{\vx\in\bbracket{q}^\ell : \wt(\vx;q^*)\in[w_1,w_2]\right\}
\] 
be the set of all sequences with $\ell$-gram weights restricted to
$[w_1,w_2]$.  
For example, representing $A,T,G,C$ by $0,1,2,3$,
respectively, and setting $q=4$ and $q^*=2$, the choice $w_1=\lfloor
\ell/2 \rfloor, \; w_2=w_1+1$ enforces the balanced GC
constraint. Also, note that $S(q,\ell;q^*,[0,\ell]) =
\bbracket{q}^\ell$, for any choice of $q^*$.

\item {\bf Forbidden $\ell$-grams}. 
Studies have indicated that certain 
substrings in DNA sequences -- such as GCG, CGC -- are likely to cause sequencing errors (see \cite{Nakamura.etal:2011}). Hence, one may also choose $S$ so as to avoid certain $\ell$-grams.
Treatment of specialized sets of forbidden $\ell$-grams is beyond the scope of this paper and 
is deferred to future work.
\end{enumerate}

Therefore, with an appropriate choice of $S$, 
we may lower the probability of substitution errors due to synthesis, lack of coverage and sequencing. 
Furthermore, as we show in our subsequent derivations, a carefully chosen set $S$ may improve the error-correcting capability 
by designing codewords to be at a sufficiently large ``distance'' from each other.
Next, we formally define the notion of sequence and profile distance as well as error-correcting codes for the corresponding DNA channel.

\subsection{Error-Correcting Codes for the DNA Storage Channel}

Fix $S\subseteq \bbracket{q}^\ell$.
Let $N$ be an integer which usually denotes the number of $\ell$-grams in the profile vector, i.e. $N=|S|$.
Define the {\em $L_1$-weight} of a word $\vu\in\ZZ^{N}_{\ge 0}$ as $\wt(\vu)\triangleq\sum_{i=1}^N u_i$.
In addition, for any pair of words $\vu,\vv\in\ZZ^{N}_{\ge 0}$, let $\Delta(\vu,\vv)\triangleq\sum_{i=1}^N \max(u_i-v_i,0)$ and define the {\em asymmetric distance} as $\da(\vu,\vv)=\max\left(\Delta(\vu,\vv), \Delta(\vv,\vu)\right)$.
A set $\C$ is called an $(N,d)$-{\em asymmetric error correcting code (AECC)} if $\C\subseteq \ZZ^{N}_{\ge 0}$ and $d=\min\{\da(\vx,\vy): \vx,\vy\in \C, \vx\ne\vy\}$. For any $\vx\in\C$, let 
$\ve\in\ZZ^N_{\ge 0}$ be such that $\vx-\ve\ge \vzero$. We say that an {\em asymmetric error $\ve$} occurred if the received word is $\vx-\ve$. We have the following theorem characterizing asymmetric error-correction codes
(see\cite[Thm 9.1]{Klove:1981}).

\begin{thm} An $(N,d+1)$-AECC corrects any asymmetric error of weight at most $d$.
\end{thm}

Next, we let $(\bbracket{q}^n;S)$ denote all $q$-ary words of length $n$ whose $\ell$-grams belong to $S$
and define the {\em $\ell$-gram distance} between two words $\vx,\vy\in(\bbracket{q}^n;S)$ as 
\[\dg(\vx,\vy;S)\triangleq \da(\vp(\vx;S),\vp(\vy;S)).\]
Note that $\dg$ is not a metric, as $\dg(\vx,\vy;S)=0$ does not imply that $\vx=\vy$. 
For example, we have $\dg(0010,1001;\bbracket{2}^2)=0$.
Nevertheless, $((\llbracket q\rrbracket^n;S),\dg)$ forms a pseudometric space. 
We convert this space into a metric space via an equivalence relation called metric identification.
Specifically, we say that $\vx\stackrel{\dg}{\sim}\vy$ if and only if $\dg(\vx,\vy;S)=0$. 
Then, by defining $\Q(n;S)\triangleq(\llbracket q\rrbracket^n;S)/\stackrel{\dg}{\sim}$, 
we can make $(\Q(n;S),\dg)$ into a metric space.
An element $X$ in $\Q(n;S)$ is an equivalence class, where 
$\vx,\vx'\in X$ implies that $\vp(\vx;S)=\vp(\vx';S)$. 
We specify the choice of {\em representative} for $X$ in Section \ref{sec:decoding}
and henceforth refer to elements in $\Q(n;S)$ by their representative words.
Let $\pQ(n;S)$ denote the set of profile vectors of words in $\Q(n;S)$. Hence, $|\pQ(n;S)|=|\Q(n;S)|$.

Let $\C\subseteq \Q(n;S)$. If $d=\min\{\dg(\vx,\vy;\ell): \vx,\vy\in \C, \vx\ne\vy\}$, then
$\C$ is called an $(n,d;S)$-{\em $\ell$-gram reconstruction code (GRC)}.
The following proposition demonstrates that an $\ell$-gram reconstruction code 
is able to correct synthesis and sequencing errors provided that its  $\ell$-gram distance is sufficiently large.
We observe that synthesis errors have effects that are $\ell$ times stronger 
since the error in some sense propagates through multiple $\ell$-grams.

\begin{prop}\label{prop:code}
An $(n,d;S)$-GRC can correct $\ssyn$ substitution errors due to synthesis, $\sseq$ substitution errors due to sequencing and $t$ undersampling errors provided that $d>2\ssyn\ell+2\sseq+t$.
\end{prop}

\begin{proof}
Consider an $(n,d;S)$-GRC $\C$ and the set
$\vp(\C)=\{\vp(\vx;S):\vx\in\C\}$. By construction, 
$\vp(\C)$ is an $(N,d)$-AECC with $N=|S|$ that 
corrects all asymmetric errors of weight $ \leq 2\ssyn\ell+2\sseq+t$.

Suppose that, on the contrary, $\C$ cannot correct $\ssyn$ substitution errors due to synthesis, $\sseq$ substitution errors due to sequencing and $t$ undersampling errors.
Then, there exist two distinct codewords  $\vx,\vx'\in C$ and error vectors $\ve_{s_1,+},\ve_{s_1,-},\ve_{s_2,+},\ve_{s_2,-},\ve_{t},\ve'_{s_1,+},\ve'_{s_1,-},\ve'_{s_2,+},\ve'_{s_2,-},\ve'_{t}$,
such that $\vhx = \vhx'$, that is, such that
\[\vx+\ve_{s_1,+}-\ve_{s_1,-}+\ve_{s_2,+}-\ve_{s_2,-}-\ve_{t}=
\vx'+\ve'_{s_1,+}-\ve'_{s_1,-}+\ve'_{s_2,+}-\ve'_{s_2,-}-\ve_{t}.\]
Here, $\ve_{s_1,-}-\ve_{s_1,+}$ and $\ve'_{s_1,-}-\ve'_{s_1,+}$ are the error vectors due to substitutions during synthesis in $\vx$ and $\vx'$, respectively; each of the vectors $\ve_{s_1,-},\ve_{s_1,+},\ve'_{s_1,-},\ve'_{s_1,+}$ has weight $\ssyn\ell$;
the vectors
$\ve_{s_2,-}-\ve_{s_2,+}$ and $\ve'_{s_2,-}-\ve'_{s_2,+}$ model substitution errors during sequencing in $\vx$ and $\vx'$, respectively; each of the vectors $\ve_{s_2,-},\ve_{s_2,+},\ve'_{s_2,-},\ve'_{s_2,+}$ has weight $\sseq$;
and $\ve_{t}$ and $\ve'_{t}$ are the undersampling error vectors of $\vx$ and $\vx'$, respectively, and both $\ve_{t},\ve'_{t}$ have weight $t$.
Therefore,
\[\vx-(\ve_{s_1,-}+\ve_{s_2,-}+\ve_{t}+\ve'_{s_1,+}+\ve'_{s_2,+})=
\vx'-(\ve'_{s_1,-}+\ve'_{s_2,-}+\ve'_{t}+\ve_{s_1,+}+\ve_{s_2,+}),\]
\noindent where
$\ve_{s_1,-}+\ve_{s_2,-}+\ve_{t}+\ve'_{s_1,+}+\ve'_{s_2,+}$ and
$\ve'_{s_1,-}+\ve'_{s_2,-}+\ve'_{t}+\ve_{s_1,+}+\ve_{s_2,+}$ are
nonnegative vectors of weight at most $2\ssyn\ell+2\sseq+t$. This
contradicts the fact that $\vx$ and $\vx'$ belong to a code that
corrects asymmetric errors with weight at most $2\ssyn\ell+2\sseq+t$.
\end{proof}

Throughout the remainder of the paper, we consider the problem of
enumerating the profile vectors in $\pQ(n;S)$ and constructing
$(n,d;S)$-$\ell$-gram reconstruction codes for a general subset
$S\subseteq\bbracket{q}^\ell$.  Our solutions are characterized by
properties associated with a class of graphs defined on $S$, which we
introduce in Section \ref{sec:graph}. In the same section, we collect
enumeration results for $\Q(n;S)$.  Section \ref{sec:enumerate} is
devoted to the proof of the main enumeration result using Ehrhart
theory.  We further exploit Ehrhart theory and certain graph theoretic
concepts to construct codes in Section \ref{sec:lower} and summarize
numerical results for the special case where $S=S(q,\ell;
q^*,[w_1,w_2])$ in Section \ref{sec:numerical}.  Finally, we describe
practical decoding procedures in Section \ref{sec:decoding}.

\begin{rem}
\begin{enumerate}[(i)] \hfill
\item For the case $S=\bbracket{q}^\ell$, given a word $\vx$, Ukkonen made certain observations 
on the structure of certain words in the equivalence class of $\vx$, 
but was unable to completely characterize all words within the class~\cite{ukkonen1992approximate}.
Here, we focus on computing the {\em number} of equivalence classes for a general subset $S$.
\item For ease of exposition, we abuse notation by identifying words in $\Q(n;S)$ 
with their corresponding profile vectors in $\pQ(n;S)$.
and refer to GRCs as being subsets of $\Q(n;S)$ or $\pQ(n;S)$ interchangeably.
\end{enumerate}
\end{rem}

\section{Restricted De Bruijn Graphs and Enumeration of Profile Vectors}
\label{sec:graph}

We use standard concepts and terminology from graph theory, following Bollob\'as~\cite{Bollobas:1998}.

A {\em directed graph (digraph)} $D$ is a pair of sets $(V,E)$, where $V$ is the set of {\em nodes}
and $E$ is a set of ordered pairs of $V$, called {\em arcs}.
If $e=(v,v')$ is an arc, we call $v$ the {\em initial} node and $v'$ the {\em terminal} node.
We allows loops in our digraphs: in other words, we allow $v=v'$.

The {\em incidence matrix} of a digraph $D$ is a matrix $\vB(D)$ in $\{-1,0,1\}^{V\times E}$, where
\begin{equation*}
\vB(D)_{v,e}=
\begin{cases}
1 & \mbox{if $e$ is not a loop and $v$ is its terminal node},\\
-1 & \mbox{if $e$ is not a loop and $v$ is its source node},\\
0 & \mbox{otherwise}.
\end{cases}
\end{equation*}
Observe that when a digraph $D$ has loops, its incidence matrix $\vB(D)$ has $\vzero$-columns indexed by these loops. When $D$ is connected, it is known that the rank of 
$\vB(D)$ equals $|V|-1$ (see \cite[\S II, Thm 9 and Ex. 38]{Bollobas:1998}).

A {\em walk} of length $n$ in a digraph is a sequence of nodes
$v_0v_1\cdots v_n$ such that $(v_i,v_{i+1})\in E$ for all
$i\in\bbracket{n}$. A walk is {\em closed} if $v_0=v_n$ and a {\em
  cycle} is a closed walk with distinct nodes, i.e., $v_i\ne v_j$, for
$0\le i<j< n$. We consider a loop to be a cycle of length one.  Given
a subset $C$ of the arc set, let $\vchi(C)\in \{0,1\}^E$ be its {\em
  incidence vector}, where $\vchi(C)_e$ is one if $e\in C$ and zero
otherwise.  In general, for any closed walk $C$ in $D$, we have
$\vB(D)\vchi(C)=\vzero$. A digraph is \emph{strongly connected} if for
all $\vz,\vz'\in V(S)$, there exists a walk from $\vz$ to $\vz'$ and
vice versa.

%

%

We are concerned with a generalization of a family of digraphs, 
namely, the de Bruijn graphs \cite{Bruijn:1946}.
Given $q$ and $\ell$, the standard {\em de Bruijn graph} is defined on the set $\llbracket q\rrbracket^\ell$.
Here, we fix a subset $S\subseteq\bbracket{q}^\ell$  and 
define the corresponding {\em restricted de Bruijn graph}, denoted by $D(S)$.
The nodes of $D(S)$ are the $(\ell-1)$-grams appearing in of words in $S$.
The pair $(\vv,\vv')$ belongs to the arc set if and only if 
$v_i=v'_{i-1}$ for $2\le i\le\ell$ and $v_1v_2\cdots v_{\ell-1}v'_{\ell-1}\in S$.
Hence, we identify the arc set with $S$ and denote the node set as $V(S)$.

The notion of restricted de Bruijn graphs was introduced by Ruskey \etal{} \cite{Ruskey.etal:2012} 
for the case of a binary alphabet. In their paper, Ruskey \etal{} showed that $D(S)$ is Eulerian 
when $S=S(2,\ell;1,[w-1,w])$ for $w\in[\ell]$.
Nevertheless, the results of \cite{Ruskey.etal:2012} can be extended for general $q$, $q^*$ and more general range of weights.
As these extensions are needed for our subsequent derivation, we provide their technical proofs in Appendix \ref{app:restricteddb}.
For purposes of brevity, we write $D(S(q,\ell; q^*,[w_1,w_2]))$ and $D(\bbracket{q}^\ell)$
 as $D(q,\ell; q^*,[w_1,w_2])$ and $D(q,\ell)$, respectively.

\begin{prop}\label{debruijn}
Fix $q$ and $\ell$.
Let $1\le q^*\le q-1$ and $1\le w_1<w_2\le \ell$.
Then $D(q,\ell;q^*,[w_1,w_2])$ is Eulerian.
In addition, $D(q,\ell)$ is Hamiltonian.
\end{prop}

Observe that when $q^*=q-1$, $w_1=0$, $w_2=\ell$, we recover the classical result that the de Bruijn graph $D(q,\ell)$ is Eulerian and Hamiltonian.

\begin{figure}
\centering
\begin{tabular}{|c|c|}
\hline
$D(2,3)$&
$\vp(0001000;2,3)$\\
\hline
\xymatrix@=20mm{
{00} \ar[d]_{001}\ar@(u,l)[]_{000} & 10\ar[l]_{100}\ar@/^/[dl]^{101}\\
01 \ar[r]_{011}\ar@/^/[ur]^{010} & 11 \ar[u]_{110}\ar@(d,r)[]_{111}
}&
\xymatrix@=20mm{
00 \ar[d]_{1}\ar@(u,l)[]_{2} & 10\ar[l]_{1}\ar@/^/[dl]^{0}\\
01 \ar[r]_{0}\ar@/^/[ur]^{1} & 11 \ar[u]_{0}\ar@(d,r)[]_{0}
}\\ \hline
\end{tabular}

\vspace{5mm}

\begin{tabular}{|c|c|}
\hline
$D(2,4;1,[2,3])$&
$\vp(011001101011;S(2,4;1,[2,3]))$\\
\hline
\xymatrix@=7mm{
001 \ar[rrr]^{0011} & *\txt{} &*\txt{} & 
011 \ar[dd]^{0110}\ar[dr]^{0111}\\
*\txt{}& 010 \ar@/^/[r]^{0101}
& 101 \ar@/^/[l]^{1010} \ar[ur]^{1011}
&*\txt{}& 111 \ar[dl]^{1110}\\
100 \ar[uu]^{1001} &*\txt{}&*\txt{}&
110 \ar[ul]^{1101}\ar[lll]^{1100}
}
&
\xymatrix@=7mm{
001 \ar[rrr]^{1} & *\txt{} &*\txt{} & 
011 \ar[dd]^{2}\ar[dr]^{0}\\
*\txt{}& 010 \ar@/^/[r]^{1}
& 101 \ar@/^/[l]^{1} \ar[ur]^{1}
&*\txt{}& 111 \ar[dl]^{0}\\
100 \ar[uu]^{1} &*\txt{}&*\txt{}&
110 \ar[ul]^{1}\ar[lll]^{1}
}\\ \hline
\end{tabular}
\caption{The restricted de Bruijn graphs defined on sets $\bbracket{2}^3$ and $S(2,4;1,[2,3])$.
We represent the respective $3$- and $4$-gram profile vectors for $0001000$ and $011001101011$
using their respective restricted de Bruijn graphs.}
\label{fig:debruijn}
\end{figure}
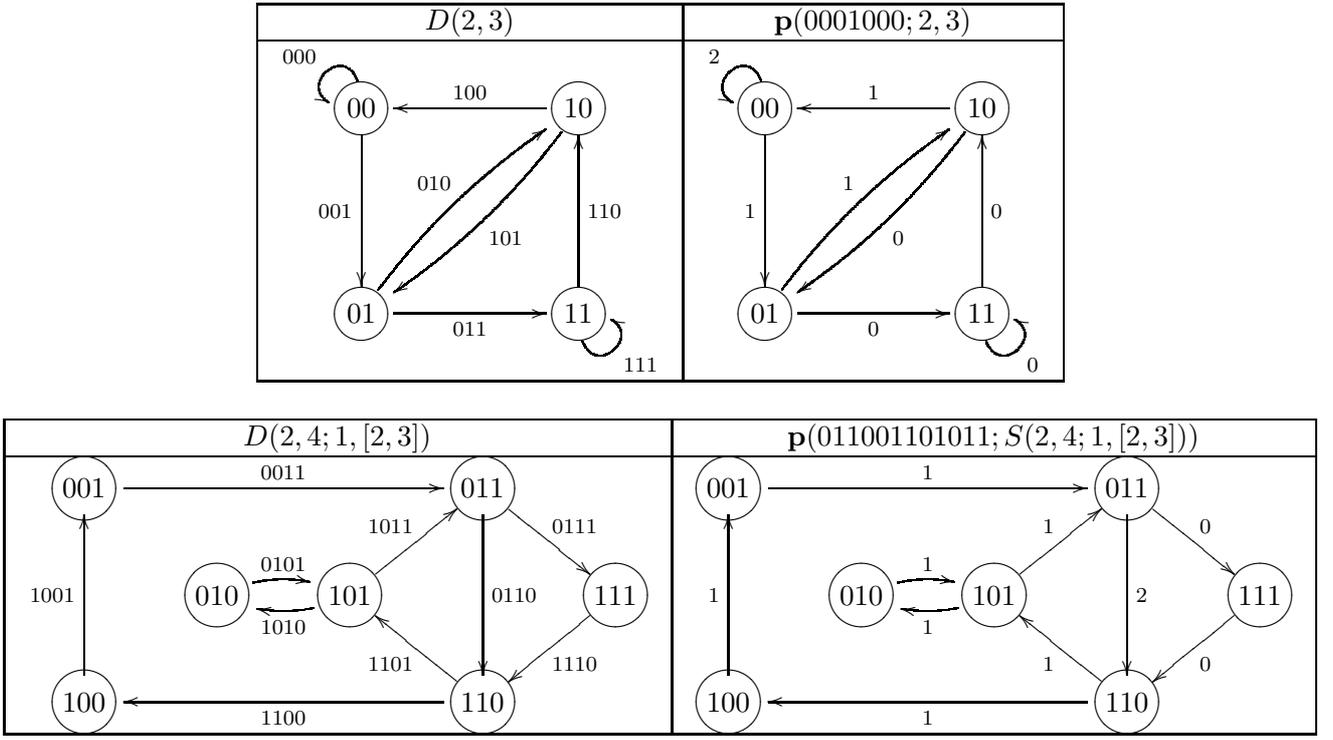

\subsection{Enumerating $\Q(n;S)$}


In this subsection, we provide the main enumeration results for $\Q(n;S)$, or equivalently, for $\pQ(n;S)$.
We first assume that $D(S)$ is strongly connected.
In addition, we consider closed walks in $D(S)$, or equivalently, {\em closed words} that start and end with the same
$(\ell-1)$-gram. We denote the set of closed words in $\Q(n;S)$ by $\Qb(n;S)$, and the corresponding set of profile vectors by $\pQb(n;S)$.

Suppose that $\vu$ belongs to $\pQb(n;S)$.
Then the following system of linear equations that we refer to as the {\em flow conservation equations},
hold true:
\begin{equation}\label{eq:flow}
\vB(D(S))\vu=\vzero.
\end{equation}

Let $\vone$ denote the all-ones vector. Since the number of $\ell$-grams in a word of length $n$ is $n-\ell+1$, 
we also have
\begin{equation}\label{eq:sum}
 \vone^T\vu=n-\ell+1.
\end{equation}

Let $\vA(S)$ be $\vB(D(S))$ augmented with a top row $\vone^T$; let $\vb$ be a vector of length $|V(S)|+1$ with a one as its first entry, and zeros elsewhere.
Equations \eqref{eq:flow} and \eqref{eq:sum} may then be rewritten as $\vA(S)\vu=(n-\ell+1)\vb$.

Consider the following two sets of integer points
{
\begin{align}
\F(n;S) & \triangleq\{\vu\in \ZZ^{|S|}: \vA(S)\vu=(n-\ell+1)\vb,\ \vu\ge\vzero\}, \label{sol_polytope}\\
\E(n;S) &\triangleq\{\vu\in \ZZ^{|S|}: \vA(S)\vu=(n-\ell+1)\vb,\ \vu>\vzero \} \label{sol_interior}.
\end{align}
} The preceding discussion asserts that the profile vector of any
closed word must lie in $\F(n;S)$.  Conversely, the next lemma shows
that any vector in $\E(n;S)$ is a profile vector of some word in
$\Qb(n;S)$.

\begin{lem}\label{lem:euler}
Suppose that $D(S)$ is strongly connected.
If $\vu\in\E(n;S)$, then there exists a word $\vx\in\Qb(n;S)$ such that $\vp(\vx;S)=\vu$. That is, $\E(n;S) \subset \pQb(n;S)$.
\end{lem}

\begin{proof}
Construct a multidigraph $D'$ on the node set $V(S)$
by adding $u_\vz$ copies of the arc $\vz$ for all $\vz\in V(S)$.
Since each $u_\vz$ is positive and $D(S)$ is strongly connected,  $D'$ is also strongly connected.
Since $\vu\in\E(n;q,\ell)$, $\vu$ also satisfies the flow conservation equations and $D'$ is consequently Eulerian. 
Also, as $D'$ has $n-\ell+1$ arcs, an Eulerian walk on $D'$ yields one such desired word $\vx$.
\end{proof}

Therefore, we have the following relation,
\begin{equation}\label{inequalities}
\E(n;S)\subseteq \pQb(n;S)\subseteq \F(n;S) .
\end{equation}

We first state our main enumeration result and defer its proof to
Section \ref{sec:enumerate}.  Specifically, under the assumption that
$D(S)$ is strongly connected, we show that both $|\E(n;S)|$ and
$|\F(n;S)|$ are quasipolynomials in $n$ whose coefficients are
periodic in $n$.  Following Beck and Robins~\cite{Beck.Robins:2007},
we define a {\em quasipolynomial} $f$ as a function in $n$ of the form
$c_D(t)t^d+c_{D-1}(t)t^{D-1}+\cdots+c_0(t)$, where
$c_D,c_{D-1},\ldots, c_0$ are periodic functions in $n$.  If $c_D$ is
not identically equal to zero, $f$ is said to be of {\em degree} $D$.
The {\em period} of $f$ is given by the lowest common multiple of the
periods of $c_D,c_{D-1},\ldots, c_0$.

In order to state our asymptotic results, we adapt the standard
$\Omega$ and $\Theta$ symbols.  We use $f(n)=\Omega'(g(n))$ to state
that for a fixed value of $\ell$, there exists an integer $\lambda$
and a positive constant $c$ so that $f(n)\ge cg(n)$ for sufficiently
large $n$ with $\lambda|(n-\ell+1)$.  In other words, $f(n)\ge cg(n)$
whenever $n$ is sufficiently large and is congruent to $\ell-1$ modulo
$\lambda$.  We write $f(n)=\Theta'(g(n))$ if $f(n)=O(g(n))$ and
$f(n)=\Omega'(g(n))$.

\begin{thm}\label{thm:closed}
Suppose $D(S)$ is strongly connected and 
let $\lambda$ be the least common multiple of the lengths of all cycles in $D(S)$.
Then $|\E(n;S)|$ and $|\F(n;S)|$ are both quasipolynomials in $n$ of the same degree $|S|-|V(S)|$ and 
share the same period that divides $\lambda$.
In particular, $|\pQb(n;S)|=\Theta'\left(n^{|S|-|V(S)|}\right)$.
\end{thm}

Before we end this section, we look at certain implications of Theorem \ref{thm:closed}.
First, we show that the estimate on $|\pQb(n;S)|$ extends to $|\pQ(n;S)|$ when $D(S)$ is strongly connected.

\begin{cor}\label{cor:notclosed}
Suppose $D(S)$ is strongly connected. For any $\vz,\vz'\in V(S)$,
 consider the set of words in $\Q(n;S)$
that begin with $\vz$ and end with $\vz'$ and let $\pQ(n;S,\vz\to\vz')$ be the corresponding set of profile vectors.
Similarly, let  $\pQ(n;S,\vz\to *)$ and $\pQ(n;S,*\to\vz')$ denote the set of profile vectors of words beginning with $\vz$ and words ending with $\vz'$, respectively.
Then 
\[
|\pQ(n;S)|=\Theta'(|\pQ(n;S,\vz\to\vz')|)=\Theta'(|\pQ(n;S,*\to\vz')|)=\Theta'(|\pQ(n;S,\vz\to*)|)=\Theta'\left(n^{|S|-|V(S)|}\right).
\]
\end{cor}

\begin{proof}
Let $\vz,\vz'\in V(S)$. 
Since $D(S)$ is strongly connected, we consider the shortest path from $\vz$ to $\vz'$ in $D(S)$.
Let $\vw=\vz\vw'$ be the corresponding $q$-ary word, $p(\vz,\vz')=|\vw'|$ be the length of the path and 
$\vu(\vz\to\vz')=\vp(\vw;S)$ be its profile vector.
Observe that both the length $p(\vz,\vz')$ and the vector  $\vu(\vz\to\vz')$
are independent of $n$.

We demonstrate the following inequality:
\begin{equation}\label{eq:notclosed} 
|\E(n-p(\vz,\vz');S)|\le |\pQ(n;S,\vz\to\vz')|\le |\pQb(n+p(\vz,\vz');S)|.
\end{equation}

First, we construct a map $\phi_1:\E(n-p(\vz,\vz');S)\to\pQ(n;S,\vz\to\vz')$ defined by $\vu\mapsto \vu+\vu(\vz\to\vz')$.
Now, since $\vu\in\E(n-p(\vz,\vz');S)$, we can assume that $\vu$ is the profile vector of a word $\vx$ 
of length $(n-p(\vz,\vz'))$ that starts and ends with $\vz$. 
Then $\vx\vw'$ is a word of length $n$ whose profile vector lies in $\pQ(n;S,\vz\to\vz')$.
Hence, $\phi_1$ is a well-defined map and it can be easily shown that the map is injective.
Therefore, the first inequality holds.

Similarly, for the other inequality, we consider the map 
$\phi_2:\pQ(n;S,\vz\to\vz')\to\pQb(n+p(\vz',\vz);S);\vu\mapsto\vu+\vu(\vz'\to\vz)$. 
As before, 
let $\vu$ be the profile vector of a word $\vx$ 
of length $n$ that starts with $\vz$ and ends with $\vz'$. 
Let  $\vw=\vz'\vw'$  be the $q$-ary word corresponding to the shortest path from $\vz'$ to $\vz$ in $D(S)$.
Concatenating $\vx$ with $\vw'$ yields $\vx\vw'$, which is a word of length $n+p(\vz',\vz)$. The profile vector 
of this word lies in $\pQb(n+p(\vz',\vz);S)$.
Hence, $\phi_2$ is a well-defined map and is injective.

Combining \eqref{eq:notclosed} with the fact that $|\E(n;S)|=\Theta'\left(n^{|S|-|V(S)|}\right)$ and 
$|\pQb(n;S)|=\Theta'\left(n^{|S|-|V(S)|}\right)$ yields the result 
$|\pQ(n;S,\vz,\vz')|=\Theta'\left(n^{|S|-|V(S)|}\right)$.

Next, we demonstrate that $|\pQ(n;S)|=\Theta'\left(n^{|S|-|V(S)|}\right)$, and observe that the other asymptotic equalities may be 
derived similarly. 
Let $P\triangleq\max\{ p(\vz,\vz'): { \vz,\vz'\in V(S)}\}$ be the diameter of the digraph $D(S)$.
Then, 
\begin{align*}
|\pQ(n;S)|  =\sum_{\vz,\vz'\in V(S)} |\Q(n;S,\vz,\vz')|
& \le \sum_{\vz,\vz'\in V(S)}|\Qb(n+p(\vz',\vz);S)|\\
& \le |V(S)|^2 |\Qb(n+P;S)|=O\left(n^{|S|-|V(S)|}\right).
\end{align*}
Since $\Q(n;S)\ge\Qb(n;S)=\Omega'\left(n^{|S|-|V(S)|}\right)$, the corollary follows.
\end{proof}

In the special case where $S=\bbracket{q}^\ell$,
Jacquet \etal{} demonstrated a stronger version of Theorem \ref{thm:closed} using analytic combinatorics.
In addition, using a careful analysis similar to the proof of Corollary \ref{cor:notclosed},
Jacquet \etal{} also provided a tighter bound for $|\pQ(n;q,\ell)|$ for the case $\ell=2$.
Note that $f(n)\sim g(n)$ stands for $\lim_{n\to\infty}f(n)/g(n)=1$.

\begin{thm}[Jacquet \etal{} \cite{Jacquet.etal:2012}]
\label{jacquet}
Fix $q,\ell$. Let $\E(n;\bbracket{q}^\ell)$, $\F(n;\bbracket{q}^\ell)$, $\pQ(n;q,\ell)$ and $\pQb(n;q,\ell)$ be defined as above.
Then
\begin{equation}\label{eq:asym}
|\E(n;\bbracket{q}^\ell)|\sim |\F(n;\bbracket{q}^\ell)|\sim |\pQb(n,q,\ell)|\sim c(q,\ell)n^{q^{\ell}-q^{\ell-1}},
\end{equation}
where $c(q,\ell)$ is a constant. Furthermore, when $\ell=2$, 
we have $|\pQ(n;q,2)|=(q^2-q+1)|\pQb(n;q,2)|(1-O(n^{-2q}))$.
\end{thm}

Next, we extend Theorem \ref{thm:closed} to provide estimates on $\Qb(n;S)$ and $\Q(n;S)$ for general $S$,
where $D(S)$ is not necessarily strongly connected. 

Given $D(S)$, let $V_1,V_2,\ldots,V_I$ be a partition of $V(S)$ such that 
the induced subgraph $(V_i,S_i)$ is strongly connected for all $i\in [I]$.
Define $\delta_i\triangleq |S_i|-|V_i|$. Then by Theorem \ref{thm:closed}, 
there are $\Theta'(n^{\delta_i})$ closed words belonging to $\Qb(n;S_i)$ and therefore, $\Qb(n;S)$.
Suppose $\Db=\max\{\delta_i: i\in I\}$. 
Then $\Qb(n;S)=\Omega'(n^\Db)$.

On the other hand, any closed word $\vx$ in $\Qb(n;S)$ corresponds to a closed walk in $D(S)$
and a closed walk in $D(S)$ must belong to some strongly connected component $(V_i,S_i)$.
In other words, $\vx$ must belong to $\Qb(n;S_i)$ for some $i\in[I]$. Hence, we have  $|\Qb(n;S)|=O(n^\Db)$. 

\begin{cor}\label{cor:general:closed}
Given $D(S)$, let $V_1,V_2,\ldots,V_I$ be a partition of $V(S)$ such that 
the induced subgraph $(V_i,S_i)$ is strongly connected for all $i\in I$.
Define $\Db\triangleq \max\{|S_i|-|V_i|: i\in I\}$.
Then $|\Qb(n;S)|=\Theta'(n^{\Db})$.
\end{cor}

\begin{exa}
Let $q=4$, $\ell=2$ $S=\{00,01,10,12,23,32,33\}$. Then $D(S)$ is as shown below. 
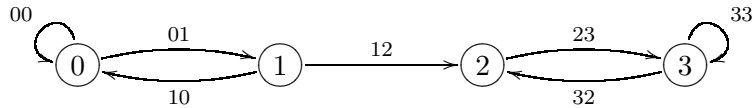
\vspace{3mm}

\centerline{
\xymatrix@=20mm{
0 \ar@(u,l)[]_{00} \ar@/^/[r]^{01} & 
1 \ar[r]^{12}\ar@/^/[l]^{10} & 
2 \ar@/^/[r]^{23} & 
3 \ar@/^/[l]^{32}\ar@(u,r)[]^{33}
}}

\vspace{3mm}
We have two strongly connected components, namely, $V_1=\{0,1\}$ and $V_2=\{2,3\}$.
So, $(V_1,S_1=\{00,01,10\})$ and $(V_2,S_2=\{23,32,33\})$ are both strongly connected digraphs with
$|\pQb(n;S_1)|=|\pQb(n;S_2)|=\floor{n/2}+1=\Theta'(n)$.
Hence, $|\pQb(n;S)|=|\pQb(n;S_1)|+|\pQb(n;S_2)|=\Theta'(n)$, in agreement with Corollary \ref{cor:general:closed}.

On the other hand, let us enumerate the elements of $\Q(n;S)$ or $\pQ(n;S)$. Let $\vu\in\pQ(n;S)$.
If $u_{12}=0$, then $\vu$ belongs to $\pQ(n;S_1)$ or $\pQ(n;S_2)$.
Otherwise, $u_{12}=1$ and  we have $\vu=\vu_1+\vchi(12)+\vu_2$ 
with $\vu_1\in \pQ(n_1;S_1, *\to 1)$, $\vu_2\in \pQ(n_2;S_2, 2\to *)$ and $n_1+n_2+1=n-\ell+1$. 

Now, $|\pQ(n;S_1)|=|\pQ(n;S_2)|=n+\ceiling{n/2}+1$ and $|\pQb(n;S_1, *\to 1)|=|\pQb(n;S_2, 2\to *)|=n+1$.
By setting $|\pQ(0;S_1, *\to 1)|=|\pQ(0;S_2, 2\to *)|=1$, we arrive
\begin{align*}
|\pQ(n;S)| &=|\pQ(n;S_1)|+|\pQ(n;S_2)|+\sum_{n_1=0}^{n-1} |\pQ(n_1;S_1,*\to 1)||\pQ(n-n_1-1;S_2,2\to *)|\\
& =2n+2\ceiling{n/2}+2+\sum_{n_1=0}^{n-1} (n_1+1)(n-n_1)\\
&= 2n+2\ceiling{n/2}+2+\frac16 n(n+1)(n+2)=\Theta'(n^3).
\end{align*}
Therefore, when $D(S)$ is not strongly connected, it is not necessarily true that $|\pQb(n;S)|$ and $|\pQ(n;S)|$ differ only by a constant factor. Furthermore, we can extend the methods in this example to obtain $|\pQ(n;S)|$
for digraphs that are not necessarily strongly connected.
 \end{exa}
 
To determine $|\pQ(n;S)|$, we construct an auxiliary weighted digraph with nodes $v_1,v_2,\ldots,v_I,v_{\rm source}$ and $v_{\rm sink}$.
If there exists an arc from the component $V_i$ to component $V_j$, $i,j\in[I]$, 
we add an arc from $v_i$ to $v_j$.
Further, we add an arc from $v_{\rm source}$ to $v_i$ and  from $v_i$ to $v_{\rm sink}$ for all $i\in[I]$.
The arcs leaving $v_{\rm source}$ have zero weight. For all $i\in [I]$, the arcs leaving $v_i$ have weight $\delta_i=|S_i|-|V_i|$ if their terminal node is $v_{\rm sink}$, and weight $\delta_i+1$ otherwise.
 (see Fig. \ref{fig:auxiliary} for the transformation).

\begin{figure}
\begin{center}
\begin{tabular}{ >{\centering\arraybackslash}m{1in} >{\centering\arraybackslash}m{1in} >{\centering\arraybackslash}m{1in}}
\xymatrix@=7mm{
*\txt{} &V_2 \ar[dr]\ar[dd] & *\txt{} \\
V_1 \ar[ur]\ar[dr] & *\txt{} & V_3 \\
*\txt{} &V_4 & *\txt{} 
}
&
$\longrightarrow$
&
\xymatrix@R=8mm@C=25mm{
*\txt{} &
v_1 \ar[d]|{\delta_1+1}\ar@/^2.5pc/[dd]|(0.25){\delta_1+1}\ar[dr]|{\delta_1} & 
*\txt{} \\
v_{\rm source}\ar[ur]|{0}\ar[r]|{0}\ar[dr]|{0}\ar[ddr]|{0} &
v_2 \ar[d]|{\delta_2+1}\ar@/_2.5pc/[dd]|(0.25){\delta_2+1}\ar[r]|{\delta_2} & 
v_{\rm sink} \\
*\txt{} &
v_3 \ar[ur]|{\delta_3} & 
*\txt{} \\
*\txt{} &
v_4 \ar[uur]|{\delta_4} & 
*\txt{} \\
}
\end{tabular}
\end{center}

\caption{Constructing a weighted digraph from the connected components of $D(S)$.}
\label{fig:auxiliary}
\end{figure}
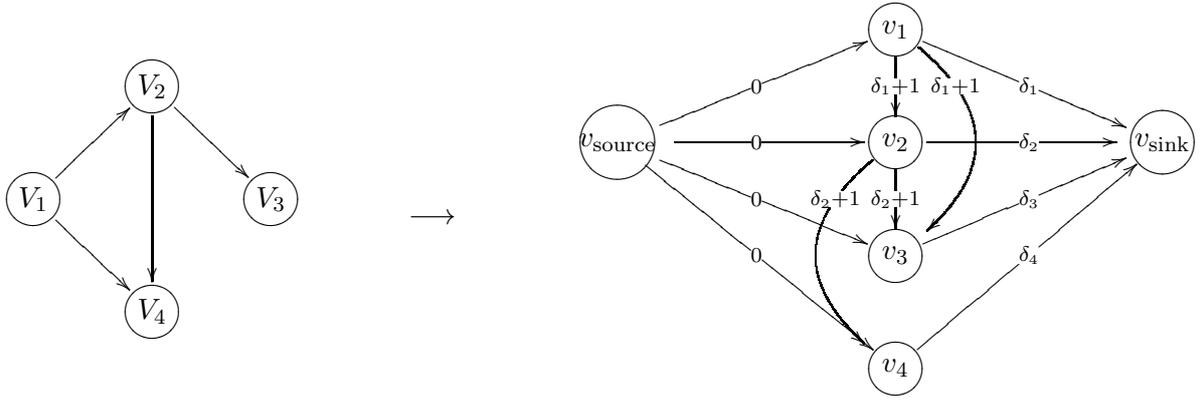

Let $D'$ be the resulting digraph and observe that $D'$ is acyclic.
Hence, we can find the longest weighted path from $v_{\rm source}$ to $v_{\rm sink}$ in linear time
(see Ahuja \etal{} \cite[Ch. 4]{Ahuja.etal:1993}). Furthermore, suppose that $\Delta$ is the weight of the longest path.
Then the next corollary states that $|\pQ(n;S)|=\Theta'(n^\Delta)$.

\begin{cor}\label{cor:generalnotclosed}
Given $D(S)$, let $V_1,V_2,\ldots,V_I$ be a partition $V(S)$ such that 
the induced subgraph $(V_i,S_i)$ is strongly connected for all $i\in I$.
Construct $D'$ as above (see Fig. \ref{fig:auxiliary}) and let $\Delta$ be the weight of the longest weighted path 
from  $v_{\rm source}$ to $v_{\rm sink}$.
Then $|\pQ(n;S)|=\Theta'(n^{\Delta})$.
\end{cor}

\begin{proof}
Let $\vu\in\pQ(n;S)$. Then there exists a set of indices $\{i_1, i_2,\ldots, i_t\}\subseteq [I]$, set of vectors 
$\vu_1,\vu_2,\ldots,\vu_{t}$, $\ve_{1},\ve_{2},\ldots, \ve_{t-1}$, and integers $n_1,n_2,\ldots,n_t$
 such that the following hold:
 \begin{itemize}
 \item $\vu=\vu_1+\ve_1+\vu_1+\ve_1+\cdots+\ve_{t-1}+\vu_t$;
 \item for $j\in [t-1]$, $\ve_j$ is the incidence vector of some arc $(\vz_j,\vz_{j+1})$ in $D(S)$ and $\vz_j\in S_{i_j}$;
 \item $\vu_1\in\pQ(n_1;S_{i_1},*\to \vz_1)$, $\vu_{t}\in\pQ(n_t;S_{i_t},\vz_{t}\to *)$ and 
 $\vu_j\in\pQ(n_j;S_{i_j},\vz_{j}\to \vz_{j+1})$ for $2\le j\le t-1$;
 \item $(t-1)+\sum_{j=1}^t n_j=n-\ell+1$;
 \item $v_{\rm source}v_{i_1}v_{i_2}\cdots v_{i_t}v_{\rm sink}$ is a path in $D'$.
 \end{itemize}
 
 For a fixed subset $\{i_1, i_2,\ldots, i_t\}\subseteq [I]$, write $n'=(n-\ell+1)-(t-1)$.
 Observe that 
 \begin{align*}
 &\sum_{\sum n_j=n'}
 |\pQ(n_1;S_{i_1},*\to \vz_1)||\pQ(n_t;S_{i_t},\vz_{t}\to *)|\prod_{j=2}^{t-1}|\pQ(n_j;S_{i_j},\vz_{j}\to \vz_{j+1})|\\
 &= \sum_{n_1=0}^{n'}\sum_{n_2=0}^{n'-n_1}\cdots \sum_{n_{t-1}=0}^{n'-n_1-\cdots-n_{t-2}} 
 O\left(n^{\delta_{i_1}+\delta_{i_1}+\cdots+\delta_{i_t}}\right)
 =O\left(n^{\delta_{i_1}+\delta_{i_1}+\cdots+\delta_{i_t}+(t-1)}\right)=O(n^\Delta).
 \end{align*}
 The last equality follows from the fact that $(t-1)+\sum_{j=1}^t \delta_{i_j}$ measures the weight of $v_{\rm source}v_{i_1}v_{i_2}\cdots v_{i_t}v_{\rm sink}$ and this value is upper bounded by $\Delta$.
 Since the number of subsets of $[I]$ is independent of $n$, we have $|\pQ(n;S)|=O(n^\Delta)$.
 
Conversely, suppose $v_{\rm source}v_{i_1}v_{i_2}\cdots v_{i_t}v_{\rm sink}$ is a path in $D'$ of maximum weight $\Delta$. Define $\vz_j$, $S_{i_j}$, $n_j$ and $n'$ as before. We then have
 \begin{align*}
 |\pQ(n;S)|
 &\ge \sum_{\sum n_j=n'}
 |\pQ(n_1;S_{i_1},*\to \vz_1)||\pQ(n_t;S_{i_t},\vz_{t}\to *)|\prod_{j=2}^{t-1}|\pQ(n_j;S_{i_j},\vz_{j}\to \vz_{j+1})|\\
 &\ge \sum_{\sum n_j=n'}
 C_1 n_1^{\delta_{i_1}}n_2^{\delta_{i_1}}\cdots n_t^{\delta_{i_t}}\\
& \ge \sum_{\sum n_j=n'} C_2 n^{\delta_{i_1}+\delta_{i_1}+\cdots+\delta_{i_t}} \mbox{   (by Jensen's inequality)}\\
 & \ge C_3 n^{\delta_{i_1}+\delta_{i_1}+\cdots+\delta_{i_t}+(t-1)}=C_3 n^\Delta,\\
 \end{align*}
\noindent where $C_1$, $C_2$ and $C_3$ are positive constants. Therefore,  $|\pQ(n;S)|=\Omega'(n^\Delta)$, completing the proof.
\end{proof}

\section{Ehrhart Theory and Proof of Theorem \ref{thm:closed}}\label{sec:enumerate}

We assume $D(S)$ to be strongly connected and provide a detailed proof
of Theorem \ref{thm:closed}.  For this purpose, In the next
subsection, we introduce some fundamental results from Ehrhart theory.
Ehrhart theory is a natural framework for enumerating profile vectors
and one may simplify the techniques of \cite{Jacquet.etal:2012}
significantly and obtain similar results for a more general family of
digraphs.  Furthermore, Ehrhart theory also allows us to extend the
enumeration procedure to profiles at a prescribed distance.

\subsection{Ehrhart Theory}

As hinted by \eqref{sol_polytope} and \eqref{sol_interior},
to enumerate codewords of interest, we need to enumerate certain sets of integer points or 
lattice points in polytopes.
The first general treatment of the theory of enumerating lattice points in polytopes was described by Ehrhart~\cite{Ehrhart:1962},
and later developed by Stanley from a commutative-algebraic point of view
(see \cite[Ch. 4]{Stanley:2011}).
Here, we follow the combinatorial treatment of Beck and Robins~\cite{Beck.Robins:2007}.

Consider any {\em rational polytope} $\PP$ given by 
\[\PP\triangleq\{\vu\in\RR^n: \vA\vu\le\vb\},\]
for some integer matrix $\vA$ and some integer vector $\vb$.
A rational polytope is {\em integer} if all its vertices are integral.
The {\em lattice point enumerator} $L_\PP(t)$ of $\PP$ is given by 
\[L_\PP(t)\triangleq |\ZZ^n\cap t\PP|,\mbox{ for all }t\in \ZZ_{>0}.\]

Ehrhart~\cite{Ehrhart:1962} introduced the lattice point enumerator for rational polytopes and 
showed that $L_\PP(t)$ is a quasipolynomial of degree $D$, 
where $D$ is given by the dimension of the polytope $\PP$.
Here, we define the {\em dimension} of a polytope to be the dimension of the affine space 
spanned by points in $\PP$.
A formal statement of Ehrhart's theorem is provided below.

\begin{thm}[{Ehrhart's theorem for polytopes \cite[Thm 3.8 and 3.23]{Beck.Robins:2007}}]
\label{thm:ehrhart}
If $\PP$ is a rational convex polytope of dimension $D$,
then $L_\PP(t)$ is a quasipolynomial of degree $D$.
Its period divides the least common multiple of 
the denominators of the coordinates of the vertices of $\PP$.
Furthermore, if $\PP$ is integer,
then $L_\PP(t)$ is a polynomial of degree $D$.
\end{thm}

Motivated by \eqref{sol_interior}, we consider the {\em relative interior} of $\PP$. 
For the case where $\PP$ is convex, the relative interior, or interior, is given by
\[\Pint \triangleq\{\vu\in\PP: \mbox{ for all $\vu'\in\PP$, there exists an $\epsilon>0$ such that $\vu+\epsilon(\vu-\vu')\in\PP$}\}.\]

For a positive integer $t$, we consider the quantity 
\[L_\Pint(t)=|\ZZ^n\cap t\Pint|.\]
Ehrhart conjectured the following relation between $L_\PP(t)$ and $L_\Pint(t)$, 
proved by Macdonald~\cite{Macdonald:1971}.

\begin{thm}[{Ehrhart-Macdonald reciprocity \cite[Thm 4.1]{Beck.Robins:2007}}]
\label{thm:ehrhar-macdonald}
If $\PP$ is a rational convex polytope of dimension $D$,
then the evaluation of $L_\PP(t)$ at negative integers satisfies
\[L_\PP(-t)=(-1)^D L_\Pint(t).\]
\end{thm}

\subsection{Proof of Theorem \ref{thm:closed}}

Recall the definitions of $\vA(S)$ and $\vb$ in \eqref{sol_polytope}, and consider the polytope
\begin{equation}
\PP(S) \triangleq\{\vu\in \RR^{|S|}: \vA(S)\vu=\vb,\vu\ge \vzero\}, \label{polytope}\\
\end{equation}

Using lattice point enumerators, we may write $|\F(n;S)|=L_{\PP(S)}(n-\ell+1)$. 
Therefore, in view of Ehrhart's theorem, we need to determine the dimension of the 
polytope $\PP(S)$ and characterize the interior and the vertices of this polytope.

\begin{lem}\label{dim:ps}
Suppose that $D(S)$ is strongly connected. 
Then the dimension of $\PP(S)$ is $|S|-|V(S)|$. 
\end{lem}

\begin{proof}
  We first establish that the rank of $\vA(S)$ is $|V(S)|$.  Since
  $D(S)$ is connected, the rank of $\vB(D(S))$ is $|V(S)|-1$.  We next
  show that $\vone^T$ does not belong to the row space of $\vB(D(S))$.
  As $D(S)$ is strongly connected, $D(S)$ contains a cycle, say $C$.
  Since $\vB(D(S))\vchi(C) = 0$ but $\vone\vchi(C)=|C|\ne 0$, $\vone$
  does not belong to the row space of $\vB(D(S))$, so augmenting the
  matrix with the all-one row increases its rank by one.  Therefore,
  the nullity of $\vA(S)$ is $|S|-|V(S)|$.

Next, we show that there exists a $\vu>\vzero$ such that $\vA(S)\vu=\vb$.
Since the nullity of $\vB(D(S))$ is positive, there exists a $\vu'$ such that  $\vA(S)\vu'=\vb$.
Since $D(S)$ is strongly connected, there exists a closed walk on $D(S)$ that visits all arcs at least once.
In other words, there exists a vector $\vv>\vzero$ such that $\vA(S)\vv=\mu\vb$ for $\mu>0$.
Choose $\mu'$ sufficiently large so that $\vu'+\mu'\vv>\vzero$ and set $\vu=(\vu'+\mu'\vv)/(1+\mu'\mu)$. 
One can easily verify that $\vA(S)\vu=\vb$.

To complete the proof, we exhibit a set of $|S|-|V(S)|+1$ affinely independent points in $\PP(S)$.
Let $\vu_1,\vu_2,\ldots$, $\vu_{|S|-|V(S)|}$ be linearly independent vectors
that span the null space of $\vA(S)$.
Since $\vu$ has strictly positive entries, we can find $\epsilon$ small enough so that 
$\vu+\epsilon\vu_i$ belongs to $\PP(S)$ for all $ i\in [|S|-|V(S)|]$. 
Therefore $\{\vu,\vu+\epsilon\vu_1,\vu+\epsilon\vu_2,\ldots, \vu+\epsilon\vu_{|S|-|V(S)|}\}$ 
is the desired set of $|S|-|V(S)|+1$ affinely independent points in $\PP(S)$.
\end{proof}

\begin{lem}\label{int:ps}
Suppose $D(S)$ is strongly connected. 
Then $\Pint(S)=\{\vu\in\RR^{|S|}: \vA(S)\vu=\vb, \vu>\vzero\}$. Therefore, $|\E(n;S)|=L_{\Pint(S)}(n-\ell+1)$.
\end{lem}

\begin{proof}
Let $\vu>\vzero$ be such that $\vA(S)\vu=\vb$.
For any $\vu'\in\PP(S)$, we have $\vA(S)\vu'=\vb$ and hence, $\vA(S)(\vu-\vu')=\vzero$.
Since $\vu$ has strictly positive entries, we choose $\epsilon$ small enough so that $\vu+\epsilon(\vu-\vu')\ge\vzero$.
Therefore, $\vu+\epsilon(\vu-\vu')$ belongs to $\PP(S)$ and $\vu$ belongs to the interior of $\PP(S)$. 

Conversely, let $\vu\in\PP(S)$, with $u_\vz=0$ for some $\vz\in S$.
Since $D(S)$ is strongly connected, from the proof of Lemma \ref{dim:ps}, there exists a $\vu'\in\PP(S)$ with $\vu'>\vzero$. Hence, for all $\epsilon>0$, the $\vz$-coordinate of $\vu+\epsilon(\vu-\vu')$ 
is given by $-\epsilon u_\vz'$, which is always negative. In other words, $\vu$ does not belong to $\Pint(S)$.
\end{proof}

Therefore, using Ehrhart's theorem and Ehrhart-Macdonald reciprocity along with Lemmas \ref{dim:ps} and \ref{int:ps}, we arrive at the fact that $|\E(n;S)|$ and $|\F(n;S)|$ are quasipolynomials in $n$ whose coefficients are periodic in $n$.

In order to determine the period of the quasipolynomials, we characterize the vertex set of $\PP(S)$.
A point $\vv$ in a polytope is a {\em vertex} if 
$\vv$ cannot be expressed as a convex combination of the other points.

\begin{lem} \label{lem:vertex-ps}
The vertex set of $\PP(S)$ is given by $\{\vchi(C)/|C|: C \mbox{ is a cycle in }D(S)\}$.
\end{lem}

\begin{proof}
First, observe that $\vchi(C)/|C|$ belongs to $\PP(S)$ for any cycle $C$ in $D(S)$. 
 
Let $\vv\in\PP(S)$ and suppose $\vv$ is a vertex.
Since $\vA(S)$ has integer entries, $\vv$ is rational. Choose $\mu > 0$ so that $\mu\vv$ has integer entries. 
Construct the multigraph $D'$ on $V(S)$ 
by adding $\mu v_\vz$ copies of the arc $\vz$ for all $\vz\in S$.
Since $\vv\in\PP(S)$, $\vB(S)\mu\vv=\vzero$ and hence, each of the connected components of $D'$ are Eulerian. Therefore, the arc set of $D'$ can be decomposed into disjoint cycles.
Since $\vv$ is a vertex, there can only be one cycle and hence, $\vv=\vchi(C)/|C|$ for some cycle $C$.

Conversely, we show that for any cycle $C$ in $D(S)$,
$\vchi(C)/|C|$ cannot be expressed as a convex combination of other points in $\PP(S)$. 
Suppose otherwise. Then there exist cycles $C_1, C_2,\ldots, C_t$ distinct from $C$
and nonnegative scalars $\alpha_1,\alpha_2,\ldots, \alpha_t$
such that $\vchi(C)=\sum_{i=1}^t\alpha_i\vchi(C_i)$. For each $j$, let $e_j$ be an arc that belongs to $C_j$
but not $C$.
Then \[0=\vchi(C)_{e_j}=\sum_{1\le i\le t}\alpha_i\vchi(C_i)_{e_j}\ge \alpha_j\vchi(C_j)_{e_j}=\alpha_j.\] 
Hence, $\alpha_j=0$ for all $j$. Therefore, $\vchi(C)=\vzero$, a contradiction.
\end{proof}

Let $\lambda_S=\lcm\{|C|: C\mbox{ is a cycle in }D(S)\}$, where $\lcm$ denotes the lowest common multiple.
Then the period of the quasipolynomial $L_{\PP(S)}(n-\ell+1)$ divides $\lambda_S$ by Ehrhart's theorem.

Let us dilate the polytope $\PP(S)$ by $\lambda_S$ and consider the polytope $\lambda_S\PP(S)$ 
and $L_{\lambda_S\PP(S)}(t)$.
Since $\lambda_S\PP$ is integer, both 
$L_{\lambda_S\PP(S)}(t)$ and $L_{\lambda_S\Pint(S)}(t)$ are polynomials of degree $|S|-|V(S)|$.
Hence,
\[ |\Qb(n;S)|\ge L_{\lambda_S\Pint(S)}(t)=\Omega\left(t^{|S|-|V(S)|}\right), 
\mbox{ whenever $n-\ell+1=\lambda_St$ or $\lambda_S|(n-\ell+1)$,} \]
\noindent and therefore, 
$|\Qb(n;S)|=\Theta'\left(n^{|S|-|V(S)|}\right)$. This completes the proof of Theorem \ref{thm:closed}.

In the special case where $D(S)$ contains a loop, we can show further that the leading coefficients of the quasipolynomials 
$|\E(n;\bbracket{q}^\ell)|$ and $|\F(n;\bbracket{q}^\ell)|$ are the same and constant. 
This result is a direct consequence of Ehrhart-Macdonald reciprocity and the fact that $|\E(n;\bbracket{q}^\ell)|$ is monotonically increasing. We demonstrate the latter claim in Appendix \ref{app:aperiodic}.

Note that when $S=\bbracket{q}^\ell$, Corollary \ref{cor:aperiodic} yields \eqref{eq:asym}, a result of Jacquet \etal{} \cite{Jacquet.etal:2012}.

\begin{cor}\label{cor:aperiodic}
Suppose $D(S)$ is strongly connected. If $D(S)$ contains a loop, then
\begin{equation}
|\E(n;S)|\sim|\Qb(n;S)|\sim|\F(n;S)|\sim c(S)n^{|S|-|V(S)|}+O(n^{|S|-|V(S)|-1}), \mbox{ for some constant }c(S).
\end{equation}
\end{cor}


\section{Constructive Lower Bounds}\label{sec:lower}

Fix $S\subseteq \bbracket{q}^\ell$ and 
recall that $\pQ(n;S)$ denotes the set of all $\ell$-gram profile vectors of words in $\Q(n;S)$.
For ease of exposition, we henceforth identify words in $\Q(n;S)$ with their corresponding profile vectors in $\pQ(n;S)$.
In Section \ref{sec:decoding}, we provide an efficient method to 
map a profile vector in $\pQ(n;S)$ back to a $q$-ary codeword in $\Q(n;S)$,
Therefore, in this section, we construct GRCs as sets of profile vectors $\pQ(n;S)$ which we may map back to corresponding $q$-ary codewords in $\Q(n;S)$.

Suppose that $\C$ is an $(N,d)$-AECC. We construct GRCs from $\C$ via the following methods:

\begin{enumerate}
\item When $N=|S|$, we intersect $\C$ with $\pQ(n;S)$ to 
obtain an $\ell$-gram reconstruction code. In other words,
we pick out the codewords in $\C$ that are also profile vectors.
Specifically, $\C\cap\pQ(n;S)$ is an $(n,d;S)$-GRC.
However, the size $|\C\cap\pQ(n;S)|$ is usually smaller than $|\C|$ and so,
we provide estimates to $|\C\cap\pQ(n;S)|$ for a classical family of AECCs 
in Section \ref{sec:intersect}.

\item When $N<|S|$, we extend each codeword in $\C$ to 
a profile vector of length $|S|$ in $\pQ(n;q,\ell)$. 
In contrast to the previous construction, we may in principle obtain an $(n,d;q,\ell)$-GRC with the same cardinality as $\C$.
However, one may not always be able to extend an arbitrary word to a profile vector. Section \ref{sec:sys} describes one method of  
mapping words in $\bbracket{m}^{N}$ to $\pQ(n;q,\ell)$ that preserves the code size for a suitable choice of 
the parameters $m$ and $N$.
\end{enumerate}

\subsection{Intersection with $\pQ(n;S)$}\label{sec:intersect}
In this section, we estimate $|\C\cap\pQ(n;S)|$ when $\C$ belongs to a
classical family of AECCs proposed by Varshamov
\cite{Varshamov:1973}. Fix $d$ and let $p$ be a prime such that $p>d$
and $p>N$.  Choose $N$ distinct nonzero elements
$\alpha_1,\alpha_2,\ldots,\alpha_N$ in $\ZZ/p\ZZ$ and consider the
matrix
\[
\vH\triangleq
\left(\begin{array}{cccc}
\alpha_1 & \alpha_2 & \cdots & \alpha_N \\
\alpha_1^2 & \alpha_2^2 & \cdots & \alpha_N^2 \\
\vdots &\vdots& \ddots &\vdots\\
\alpha_1^d & \alpha_2^d & \cdots & \alpha_N^d 
\end{array}\right).
\]
Pick any vector $\vbeta\in(\ZZ/p\ZZ)^N$ and define the code
\[\C(\vH,\vbeta)\triangleq\{\vu:\vH\vu\equiv\vbeta \bmod p\}.\]
Then, $\C(\vH,\vbeta)$ is an $(N,d+1)$-AECC \cite{Varshamov:1973}.
Hence, $\C(\vH,\vbeta)\cap \pQ(n;S)$ is an $(n,d+1;S)$-GRC for all $\vbeta\in(\ZZ/p\ZZ)^N$.
Therefore, by the pigeonhole principle, there exists a $\vbeta$ such that $|\C(\vH,\vbeta)\cap \pQ(n;S)|$
is at least $|\pQ(n;S)|/p^d$.
However, the choice of $\vbeta$ that guarantees this lower bound is not known.

In the rest of this section, we fix a certain choice of $\vH$ and $\vbeta$ and 
provide lower bounds on the size of $\C(\vH,\vbeta)\cap \pQ(n;S)$ 
as a function of $n$. 
As before, instead of looking at $\pQ(n;S)$ directly, 
we consider the set of closed words  $\Qb(n;S)$ and 
the corresponding set of profile vectors  $\pQb(n;S)$.

Let $\vbeta=\vzero$ and choose $\vH$ and $p$ based on the restricted de Bruijn digraph $D(S)$.
For an arbitrary matrix $\vM$, let $\nullplus\vM$ denote the set of vectors in the null space of $\vM$
that have positive entries.
We assume $D(S)$ to be strongly connected so that $\nullplus\vB(D(S))$ is nonempty.
Hence, we choose $\vH$ and $p$ such that $\C(\vH,\vzero)\cap \nullplus\vB(D(S))$ is nonempty.

Define the $(|V(S)|+1+d)\times (|S|+d)$-matrix
\[
\vA(\vH,S)\triangleq
\left(\begin{array}{c|c}
\vA(S) & \vzero\\ \hline
\vH &-p\vI_d
\end{array}\right),
\]
\noindent where $\vA(S)$ is as described in Section \ref{sec:graph}.
Let $\vb$ be a vector of length $|V(S)|+1+d$ that has 
$1$ as the first entry and zeros elsewhere, and define the polytope
\begin{equation}\label{eq:pgrc}
\Pgrc(\vH,S)\triangleq\{\vu\in\RR^{|S|+d}: \vA(\vH,S)\vu=\vb, \vu\ge\vzero\}
\end{equation}

Since $\E(n;S)\subseteq \pQb(n;S)\subseteq \pQ(n;S)$,  
$|\C(\vH,\vzero)\cap \E(n;S)|$ is a lower bound for  $|\C(\vH,\vzero)\cap \pQ(n;S)|$.
The following proposition demonstrates that $|\C(\vH,\vzero)\cap \E(n;S)|$
is given by the number of lattice points in the interior of a dilation of $\Pgrc(\vH,S)$.

\begin{prop} \label{codepolytope}
Let $\C(\vH,\vzero)$ and $\Pgrc(\vH,S)$ be defined as above.
If $D(S)$ is strongly connected and $\C(\vH,\vzero)\cap \nullplus\vB(D(S))$ is nonempty,
then $|\C(\vH,\vzero)\cap \E(n;S)|=\left\lvert\ZZ\cap (n-\ell+1)\Pgrcint(\vH,S)\right\rvert$.
\end{prop}

\begin{proof}
Similar to Lemma \ref{int:ps}, we have that $\Pgrcint(\vH,S)=\{\vu\in\RR^{|S|+d}: \vA(\vH,S)\vu=\vb, \vu>\vzero\}$,
and we defer the proof of this claim to Appendix \ref{app:pgrc}.

Let $\vu>\vzero$ be such that $\vA(\vH,S)\vu=(n-\ell+1)\vb$. 
Consider the vector $\vu_0$ which equals the vector $\vu$ restricted to the first $N$ coordinates.
Then $\vA(S)\vu_0=(n-\ell+1)\vb_0$, 
where $\vb_0$ is a vector of length $|V(S)|+1$ with one in its first coordinate and zeros elsewhere. 
Hence, $\vu_0\in\E(n;S)$.
On the other hand, $\vH\vu_0=p\vbeta'$, where $\vbeta'$ consists of the last $d$ entries of $\vu$. In other words, $\vH\vu_0\equiv\vzero \bmod p$ 
and so $\vu_0\in\C(\vH,\vzero)$.

Therefore, $\vu\mapsto \vu_0$ is a map from $\{\vu: \vA(\vH,S)\vu=(n-\ell+1)\vb\mbox{ and }\vu>\vzero\}$ 
to $\C(\vH,\vzero)\cap \E(n;q,\ell)$. It can be verified that this map is a bijection. This proves the claimed result.
\end{proof}

As before, we compute the dimension of $\Pgrc(\vH,S)$ and characterize its vertex set.
Since the proofs are similar to the ones in Section \ref{sec:enumerate},
the reader is referred to Appendix \ref{app:pgrc} for a detailed analysis.

\begin{lem}\label{lem:pgrc}
Let $\C(\vH,\vzero)$ and $\Pgrc(\vH,S)$ be defined as above.
Suppose further that $D(S)$ is strongly connected and $\C(\vH,\vzero)\cap \nullplus\vB(D(S))$ is nonempty.
The dimension of $\Pgrc(\vH,S)$ is $|S|-|V(S)|$, while its vertex set is given by
\[\left\{\left(\frac{\vchi(C)}{|C|},\frac{\vH\vchi(C)}{p|C|}\right): C \mbox{ is a cycle in }D(S)\right\}.\]
\end{lem}

Let $\lambda_{\rm GRC}=\lcm\{|C|: C\mbox{ is a cycle in }D(S)\}\cup\{p\}$.
Then Lemma \ref{lem:pgrc}, Ehrhart's theorem and Ehrhart-Macdonald's reciprocity imply that 
$L_{\Pgrcint(\vH,S)}(t)$ is a quasipolynomial of degree $|S|-|V(S)|$ whose period divides $\lambda_{\rm GRC}$.
As in Section \ref{sec:enumerate}, we dilate the polytope $\Pgrc(\vH,S)$ by $\lambda_{\rm GRC}$ to 
obtain an integer polytope and assume that the polynomial 
$L_{\lambda_{\rm GRC}\Pgrc(\vH,S)}(t)$ has leading coefficient $c$.
Hence, whenever $n-\ell+1=\lambda_{\rm GRC} t$, that is, whenever $\lambda_{\rm GRC}|(n-\ell+1)$, 
\begin{align*}
|\C(\vH,\vzero)\cap \E(n;S)|= L_{\lambda_{\rm GRC}\Pgrcint(\vH,S)}(t)
&=ct^{|S|-|V(S)|}+O(t^{|S|-|V(S)|-1})\\
&=c(n/\lambda_{\rm GRC})^{|S|-|V(S)|}+O(n^{|S|-|V(S)|-1}).
\end{align*}
We denote $c/\lambda_{\rm GRC}^{|S|-|V(S)|}$ by $c(\vH,S)$ and
summarize the results in the following theorem.


\begin{thm}\label{thm:code}
Fix $S\subseteq \bbracket{q}^\ell$ and $d$. 
Choose $\vH$ and $p$ so that 
$\C(\vH,\vzero)$ is an $(|S|,d+1)$-AECC and 
$\C(\vH,\vzero)\cap \nullplus\vB(D(S))$ is nonempty.
Suppose that
$\lambda_{\rm GRC}=\lcm\{{\{|C|: C\mbox{ is a cycle in }D(S)\}\cup\{p\}\}}$.
Then there exists a constant $c(\vH,S)$ such that whenever $\lambda_{\rm GRC}|(n-\ell+1)$,
\[
|\C(\vH,\vzero)\cap\pQ(n;S)| \ge c(\vH,S)n^{|S|-|V(S)|}+O(n^{|S|-|V(S)|-1}).
\]
\end{thm}

Theorem \ref{thm:code} guarantees that the code size is at least $c(\vH,S)n^{|S|-|V(S)|}$
for some constant $c(\vH,S)$.
In other words, when $d$ is constant, we have $C(n,d;S)=\Omega'(n^{|S|-|V(S)|})$.
Since $C(n,d;S)\le |\Q(n;S)|=O(n^{|S|-|V(S)|})$, we have $C(n,d;S)=\Theta'(n^{|S|-|V(S)|})$.

\subsection{Systematic Encoding of Profile Vectors}\label{sec:sys}

In this subsection, we look at efficient one-to-one mappings from 
$\llbracket m\rrbracket^{N}$ to $\pQ(n;S)$.
As with usual constrained coding problems, 
we are interested in maximizing the number of messages, i.e. the size of $m^N$,
so that the number of messages is close to $|\pQ(n;S)|=\Theta'(n^{|S|-|V(S)|})$.
We achieve this goal by exhibiting a systematic encoder with $m=\Theta(n)$ and $N=|S|-|V(S)|-1$.
More formally, we prove the following theorem.

\begin{thm}[Systematic Encoder] \label{thm:sys}
Fix $n$ and $S\subseteq \bbracket{q}^\ell$. 
Pick any $m$ so that
\begin{equation}\label{eq:sys}
m\le \frac{n-\ell+1}{\binom{|V(S)|}{2}(q-1)+|S|-|V(S)|-1}.
\end{equation}
Suppose further that $D(S)$ is Hamiltonian and contains a loop.
Then, there exists a set $I\subseteq S$ of coordinates of size
$|S|-|V(S)|-1$ with the following property:
for any $\vv\in\llbracket m\rrbracket^{I}$, 
there exists an $\ell$-gram profile vector $\vu\in\pQ(n;S)$
such that $\vu|_I=\vv$.
Furthermore, $\vu$ can be found in time $O(|V(S)|)$.
\end{thm} 

In other words, given any word $\vv$ of length $N=|I|=|S|-|V(S)|-1$, 
one can always extend it to obtain a profile vector $\vu\in \pQ(n;S)$ of length $|S|$.
As pointed out earlier, this theorem provides a simple way of constructing $\ell$-gram codes from AECCs and we sketch the construction in what follows.

Let $\phi_{\rm sys}(\vv)$ denote the profile vector resulting from Theorem \ref{thm:sys} given input $\vv$.
Consider an $m$-ary $(N,d)$-AECC $\C$ with $N=|S|-|V(S)|-1$ and $m$ satisfying \eqref{eq:sys}.
Let $\phi_{\rm sys}(\C)\triangleq \{\phi_{\rm sys}(\vv): \vv\in \C\}$.
Then $\phi_{\rm sys}(\C)\subseteq \pQ(n;S)$. Furthermore, $\phi_{\rm sys}(\C)$  has asymmetric distance at least $d$ 
since restricting the code $\phi_{\rm sys}(\C)$ on the coordinates in $I$ yields $\C$.
Hence, we have the following corollary.

\begin{cor}\label{cor:sys}
Fix $n$ and $S\subseteq \bbracket{q}^\ell$ and pick $m$ satisfying \eqref{eq:sys}.
Suppose $D(S)$ is Hamiltonian and contains a loop.
If  $\C$ is an $m$-ary $(|S|-|V(S)|-1,d)$-AECC, then
$\phi_{\rm sys}(\C)\triangleq \{\phi_{\rm sys}(\vv): \vv\in \C\}$
is a $(n,d;S)$-GRC.
%
\end{cor}

For compactness, we write $V$, $\vA$ and $\vB$, instead of $V(S)$, $\vA(S)$ and $\vB(D(S))$.
To prove Theorem \ref{thm:sys}, consider the restricted de Bruijn digraph $D(S)$.
By the assumptions of the theorem, denote the set of $|V|$ arcs in a Hamiltonian cycle as $H$
and the arc corresponding to a loop by $\va_0$. 
We set $I$ to be $S\setminus (H\cup \{\va_0\})$.

We reorder the coordinates so that the arcs in $H$ are ordered first, 
followed by the arc $\va_0$ and then the arcs in $I$.
So, given $\vv=(v_1,v_2,\ldots, v_{|I|})\in\bbracket{m}^{|I|}$,
the proof of Theorem \ref{thm:sys} essentially reduces to finding integers 
 $x_1,x_2,\ldots, x_{|V|},y$ 
such that 
\begin{equation}
  \label{eq:sysvector}
  \vA\left(x_1,x_2,\ldots, x_{|V|},y,v_1,v_2,\ldots,v_{|I|}\right)^T=(n-\ell+1)\vb.  
\end{equation}
Considering the first row of $\vA$ separately from the remaining rows, we see
that \eqref{eq:sysvector} is equivalent to the following system of equations:
\begin{align}
  \sum_{i=1}^{|V|} x_i+y &=(n-\ell+1)-\sum_{i=1}^{|I|} v_i, \label{sumflow} \\
  \vzero=\vB\left(\begin{array}{c} x_1\\ \vdots \\ x_{|V|} \\ y \\ u_1 \\ \vdots \\ u_{|I|}\end{array}\right) &=
  \vB\left(\begin{array}{c} x_1\\ \vdots \\ x_{|V|} \\ 0 \\ 0 \\ \vdots \\ 0\end{array}\right) +
  \vB\left(\begin{array}{c} 0\\ \vdots \\ 0 \\ y \\ 0 \\ \vdots \\ 0\end{array}\right) +
  \vB\left(\begin{array}{c} 0\\ \vdots \\ 0 \\ 0 \\ u_1 \\ \vdots \\ u_{|I|}\end{array}\right).\label{flowconstraint0}
\end{align}
Since the first $|V|$ columns of $\vB$ correspond to the arcs in $H$, we have 
\[ 
\vB\left( x_1, \ldots , x_{|V|} , 0 , 0 , \ldots , 0\right)^T =
\left(\begin{array}{c} x_2-x_1\\ x_3-x_2 \\ \vdots \\ x_1-x_{|V|} \end{array}\right).\]
Since the $(|V|+1)$-th column of $\vB$ is a $\vzero$-column, we have 
$\vB\left( 0, \ldots , 0 , y , 0 , \ldots , 0\right)^T =\vzero$ for any $y$.

For the final summand, let $\vB\left( 0, \ldots , 0 , 0 , v_1 , \ldots , v_{|I|}\right)^T =(r_1,r_2,\ldots,r_{|V|})^T$.
We can then rewrite \eqref{flowconstraint0} as
\begin{equation}\label{flowconstraint}
x_i-x_{i+1}=r_i, \mbox{ for } 1\le i \le |V|-1.
\end{equation}
Since $\vone^T\vB=\vzero^T$, we have 
$\vone^T(r_1,r_2,\ldots,r_{|V|})^T=\sum_{i=1}^{|V|} r_i=0$. 
Furthermore, we assume without loss of generality 
that $\sum_{i=1}^j r_i\ge 0$, for all $1\le j\le |V|$.
This can be achieved by cyclically relabelling the nodes.

It suffices to show that an integer solution for 
\eqref{flowconstraint} and \eqref{sumflow} exists, satisfying $y \geq 1$ and $x_i\ge 1$ for $i\in  [|V|]$.
Consider the following choices of $x_i$ and $y$:
\begin{align*}
  x_i &= 1 + \sum_{j=1}^{i-1}r_j, \\
  y &= (n-\ell+1)-\sum_{i=1}^{|I|} v_i - \sum_{i=1}^{|V|}x_i.
\end{align*}
Clearly, $x_i$ and $y$ satisfy \eqref{flowconstraint} and
\eqref{sumflow}. Since each $v_i$ is an integer, all $r_i$ are
integers, so $x_i$ and $y$ are also integers. Furthermore, each $x_i
\geq 1$, since we chose the labeling so that $\sum_{j=1}^{i-1}r_j \geq
0$. We still must show that $y \geq 1$.

First, we observe that $r_i< (q-1)m$ for all $i$, since each vertex has at most $(q-1)$ incoming arcs in $I$ and 
by design, each $v_i$ is strictly less than $m$. Thus, each $x_i$ satisfies
\[ x_i < 1 + (i-1)(q-1)m. \]
Summing over all $i$, we have
\[ \sum_{i=1}^{|V|}x_i \leq \sum_{i=1}^{|I|}(i-1)(q-1)m = (q-1)m{\sizeof{V} \choose 2}. \]
Since also each $v_i \leq m$, we have
\[ y \geq (n-\ell+1) - m\left[\sizeof{I} + (q-1){\sizeof{V} \choose 2}\right]. \]
By the choice of $m$, it follows that $y\ge 0$. This completes the proof of Theorem~\ref{thm:sys}.

\begin{exa}\label{exa:systematic}
Let $S=\bbracket{2}^3$ and let $n=20$. Then Theorem \ref{thm:sys} states that there is a systematic encoder
that maps words from $\llbracket 2\rrbracket^{3}$ into $\pQ(20;2,3)$.
We list all eight encoded profile vectors with their systematic part highlighted in boldface.

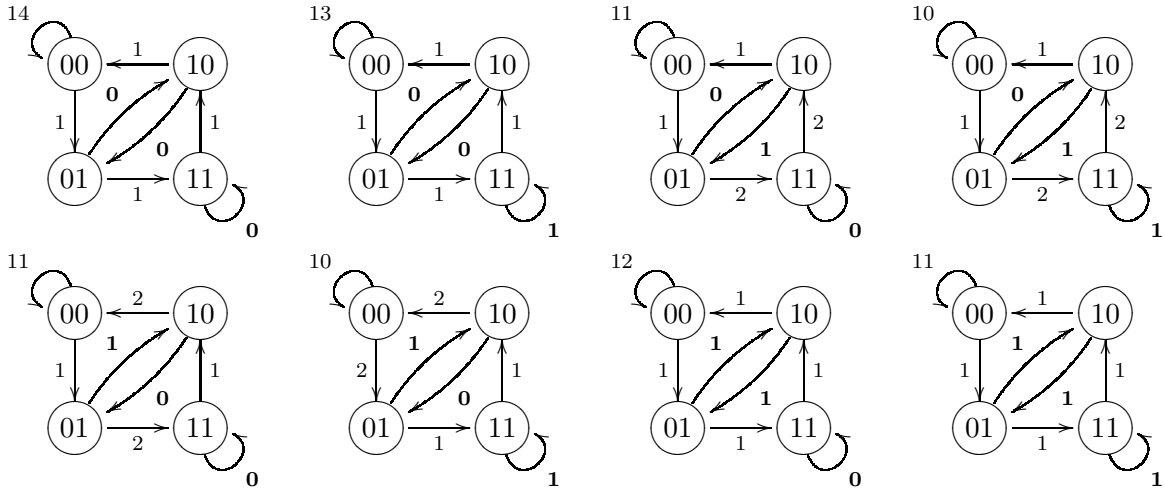
\begin{center}
\begin{tabular}{cccc}
\xymatrix@=8mm{
00 \ar[d]_{1}\ar@(u,l)[]_{14} & 10\ar[l]_{1}\ar@/^/[dl]^{{\bf0}}\\
01 \ar[r]_{1}\ar@/^/[ur]^{{\bf0}} & 11 \ar[u]_{1}\ar@(d,r)[]_{{\bf0}}
} &
\xymatrix@=8mm{
00 \ar[d]_{1}\ar@(u,l)[]_{13} & 10\ar[l]_{1}\ar@/^/[dl]^{{\bf0}}\\
01 \ar[r]_{1}\ar@/^/[ur]^{{\bf0}} & 11 \ar[u]_{1}\ar@(d,r)[]_{{\bf1}}
} &
\xymatrix@=8mm{
00 \ar[d]_{1}\ar@(u,l)[]_{11} & 10\ar[l]_{1}\ar@/^/[dl]^{{\bf1}}\\
01 \ar[r]_{2}\ar@/^/[ur]^{{\bf0}} & 11 \ar[u]_{2}\ar@(d,r)[]_{{\bf0}}
} &
\xymatrix@=8mm{
00 \ar[d]_{1}\ar@(u,l)[]_{10} & 10\ar[l]_{1}\ar@/^/[dl]^{{\bf1}}\\
01 \ar[r]_{2}\ar@/^/[ur]^{{\bf0}} & 11 \ar[u]_{2}\ar@(d,r)[]_{{\bf1}}
} \\

\xymatrix@=8mm{
00 \ar[d]_{1}\ar@(u,l)[]_{11} & 10\ar[l]_{2}\ar@/^/[dl]^{{\bf0}}\\
01 \ar[r]_{2}\ar@/^/[ur]^{{\bf1}} & 11 \ar[u]_{1}\ar@(d,r)[]_{{\bf0}}
} &
\xymatrix@=8mm{
00 \ar[d]_{2}\ar@(u,l)[]_{10} & 10\ar[l]_{2}\ar@/^/[dl]^{{\bf0}}\\
01 \ar[r]_{1}\ar@/^/[ur]^{{\bf1}} & 11 \ar[u]_{1}\ar@(d,r)[]_{{\bf1}}
} &
\xymatrix@=8mm{
00 \ar[d]_{1}\ar@(u,l)[]_{12} & 10\ar[l]_{1}\ar@/^/[dl]^{{\bf1}}\\
01 \ar[r]_{1}\ar@/^/[ur]^{{\bf1}} & 11 \ar[u]_{1}\ar@(d,r)[]_{{\bf0}}
} &
\xymatrix@=8mm{
00 \ar[d]_{1}\ar@(u,l)[]_{11} & 10\ar[l]_{1}\ar@/^/[dl]^{{\bf1}}\\
01 \ar[r]_{1}\ar@/^/[ur]^{{\bf1}} & 11 \ar[u]_{1}\ar@(d,r)[]_{{\bf1}}
} 
\end{tabular}
\end{center}

For instance, the codeword $000\in\bbracket{2}^3$ is mapped to the profile vector
$(18,1,{\bf 0},1,1,{\bf 0}, 1,{\bf 0})$.
Via the {\sc Euler} map described in Section \ref{sec:decoding}, 
this profile vector is mapped to $00\cdots01100\in\Q(24,2,3)$.

Observe that we can systematically encode $\llbracket 2\rrbracket^3$ into $\pQ(n;2,3)$
even when $n$ is smaller than 24.
In fact, in this example, we can systematically encode $\llbracket 2\rrbracket^{3}$ into $\pQ(10;2,3)$.
In general, we can can systematically encode $\llbracket m\rrbracket^{3}$ into $\pQ(4m+2;2,3)$.
In this case, the size of the message set is approximately $n^3/8$ while
the number of all possible closed profile vectors is approximately $n^4/288$~\cite{Jacquet.etal:2012}.

%

In Section \ref{sec:numerical} and Example \ref{exa:intersect}, we observe that
the construction given in Section \ref{sec:intersect} yields a larger code size.
Nevertheless, the systematic encoder is conceptually simple and 
furthermore, the systematic property of the construction in Section \ref{sec:sys} 
can be exploited to integrate rank modulation codes into our coding schemes for DNA storage, useful for
automatic decoding via \emph{hybridization}. We describe this procedure in detail in Section \ref{sec:rankmod}.
\end{exa}

\section{Numerical Computations for $S=S(q,\ell;q^*,[w_1,w_2])$} \label{sec:numerical}

In what follows, we summarize numerical results for code sizes pertaining to the special case when $S=S(q,\ell;q^*,[w_1,w_2])$.

By Proposition \ref{debruijn}, $D(q,\ell;q^*,[w_1,w_2])$ is Eulerian and therefore strongly connected.
In other words, Theorem \ref{thm:closed} applies and we have $|\Q(n;S)|=\Theta'(n^{|S|-|V(S)|})$,
where $|S|$ is given by $|S(q,\ell;q^*,[w_1,w_2])|=\sum_{w=w_1}^{w_2}\binom{\ell}{w}(q^*)^w(q-q^*)^{\ell-w}$,
while $|V(S)|$ is given by $|S(q,\ell-1;q^*,[w_1-1,w_2])|=\sum_{w=w_1-1}^{w_2}\binom{\ell-1}{w}(q^*)^w(q-q^*)^{\ell-1-w}$.

Let $D=|S|-|V(S)|$. We determine next the coefficient of $n^D$ in $|\Q(n;S)|$.
When $w_2=\ell$, the digraph $D(q,\ell;q^*,[w_1,\ell])$ 
contains the loop that correspond to the $\ell$-gram $\vone^T$.
Hence, by Corollary \ref{cor:aperiodic}, the desired coefficient is constant and 
we denote it by $c(q,\ell;q^*,[w_1,\ell])$. 
When $S=\bbracket{q}^\ell$, we denote this coefficient by $c(q,\ell)$ and 
remark that this value corresponds to the constant defined in Theorem~\ref{jacquet}.

When $w_2<\ell$, the digraph $D(q,\ell;q^*,[w_1,w_2])$ does not contain any loops.
Recall from Section \ref{sec:enumerate} the definitions of $\PP(S)$, $\lambda_S$ and $L_{\PP(S)}(n-\ell+1)$.
In particular, recall that the lattice point enumerator $L_{\PP(S)}(n-\ell+1)$ is a quasipolynomial of degree $D$
whose period divides $\lambda_S$ and that consequently, the coefficient of $n^D$ in $|\Q(n;S)|$ is periodic.
For ease of presentation, we only determine the coefficient of $n^D$
for those values for which $\lambda_S$ divides $(n-\ell+1)$ or $n-\ell+1=\lambda_St$ for some integer $t$.
In this instance, the desired coefficient is given by $c(q,\ell;q^*,[w_1,w_2])\triangleq c/\lambda_S^D$,
where $c$ is the leading coefficient of the polynomial $L_{\lambda_S\PP(S)}(t)$.

In summary, we have the following corollary.

\begin{cor}\label{cor:cql}
Consider $S=S(q,\ell;q^*,[w_1,w_2])$ and define 

$$D=\sum_{w=w_1}^{w_2}\binom{\ell}{w}(q^*)^w(q-q^*)^{\ell-w}-
\sum_{w=w_1-1}^{w_2}\binom{\ell-1}{w}(q^*)^w(q-q^*)^{\ell-1-w}.$$ 

Suppose that $\lambda_S=\lcm\{|C|: C \mbox{ is a cycle in }D(S)\}$.
Then for some constant $c(q,\ell;q^*,[w_1,w_2])$,
\begin{enumerate}
\item If $w_2=\ell$, $|\Q(n;S)|=c(q,\ell;q^*,[w_1,\ell])n^D+O(n^{D-1})$ for all $n$;
\item Otherwise, if $w_2<\ell$,  $|\Q(n;S)|=c(q,\ell;q^*,[w_1,w_2])n^D+O(n^{D-1})$ for all $n$ such that 
$\lambda_S|(n-\ell+1)$.
\end{enumerate}
When $S=\bbracket{q}^\ell$, we write $c(q,\ell)$ instead of $c(q,\ell;1,[0,\ell])$.
\end{cor}


We determine $c(q,\ell;q^*,[w_1,w_2])$ via numerical computations.
Computing the lattice point enumerator is a fundamental problem in discrete optimization and 
many algorithms and software implementations have been developed for such purposes.
We make use of the software {\tt LattE}, developed by Baldoni \etal{} \cite{Baldoni.etal:2014},
which is based on an algorithm of Barvinok~\cite{Barvinok:1994}. 
Barvinok's algorithm essentially triangulates the supporting cones of the vertices
of a polytope to obtain simplicial cones and then decompose the simplicial cones recursively
into unimodular cones. As the rational generating functions of the
resulting unimodular cones can be written down easily, 
adding and subtracting them according to the inclusion-exclusion principle and 
Brion's theorem gives the desired rational generating function of the polytope.
The algorithm is shown to enumerate the number of lattice points in polynomial time
when the dimension of the polytope is fixed.

Using {\tt LattE}, we computed the desired coefficients for various values of $(q,\ell;q^*,[w_1,w_2])$.
As an illustrative example, {\tt LattE} determined $c(2,4)= 283/9754214400$ with computational time less than a minute. This shows that although the exact evaluation of $c(q,\ell)$ is prohibitively complex (as pointed by Jacquet \etal{} \cite{Jacquet.etal:2012}), numerical computations of $c(q,\ell)$ and $c(q,\ell;q^*,[w_1,w_2])$
are feasible for certain moderate values of parameters. We tabulate these values in Table \ref{tab:cql} and \ref{tab:cqlw}.

\begin{table}[!t]
\caption{Computation of $c(q,\ell)$}
\label{tab:cql}
\begin{center}
\begin{tabular}{cccc}\hline
$q$ & $\ell$ & $D$ & $c(q,\ell)$  \\\hline
2 & 2 & 2 & 1/4* \\
3 & 2 & 6 & 1/8640* \\
4 & 2 & 12 & 1/45984153600* \\
5 & 2 & 20 & 37/84081093402584678400000* \\
2 & 3 & 4 & 1/288* \\
3 & 3 & 18 & 887/358450977137334681600000 \\
2 & 4 & 8 & 283/9754214400 \\
2 & 5 & 16 & 722299813/94556837526637331349504000000 \\
\hline
\end{tabular}
\vspace{2mm}

Entries marked by an asterisk refer to values that were also derived by Jacquet \etal{} \cite{Jacquet.etal:2012}.
\end{center}
\end{table}

\begin{table}[!t]
\caption{Computation of $c(q,\ell;q^*,[w_1,w_2])$. We fixed $q=2$ and $q^*=1$.}
\label{tab:cqlw}
\begin{center}
\begin{tabular}{cccccc}\hline
$\ell$ & $w_1$ & $w_2$ & $D$ & $\lambda_S$ & $c(2,\ell;1,[w_1,w_2])$  \\\hline
4 & 2 & 3 & 3 & 60 & 1/360 \\
4 & 2 & 4 & 4 & -- & 1/1440 \\
5 & 2 & 3 & 6 & 120 & 1/5184000 \\
5 & 2 & 4 & 10 & 27720 & 40337/34566497280000000 \\
5 & 2 & 5 & 11 & -- & 3667/34566497280000000 \\
5 & 3 & 4 & 4 & 420 & 23/302400 \\
5 & 3 & 5 & 5 & -- & 23/1512000 \\
6 & 3 & 4 & 10 & 65520 & 43919/754932300595200000 \\
6 & 3 & 5 & 15 & 5354228880 & 1106713336565579/739506679855711968646397952000000000 \\
6 & 4 & 5 & 5 & 840 & 1/518400 \\
\hline
\end{tabular}
\end{center}
\end{table}

\subsection{Lower Bounds on Code Sizes}

Next, we provide numerical results for lower bounds on the code sizes derived in Section \ref{sec:intersect}.

When $S=S(q,\ell;q^*,[w_1,w_2])$, the digraph $D(S)$ is Eulerian by Proposition \ref{debruijn}
and hence, $\vone$ belongs to ${\rm Null}_{>0}\vB(D(S))$.
Therefore, if $\C(\vH,\vzero)$ contains the vector $\vone$ as well, 
$\C(\vH,\vzero)\cap {\rm Null}_{>0}\vB(D(S))$ is nonempty and 
the condition of Theorem \ref{thm:code} is satisfied. 
Hence, we have the following corollary.

\begin{cor}\label{cor:cHS}
Let $S=S(q,\ell;q^*,[w_1,w_2])$. 
Fix $d$ and choose $\vH$ and $p$ such that $\C(\vH,\vzero)$
is an $(|S|,d+1)$-AECC containing $\vone$.
Suppose that
$\lambda_{\rm GRC}=\lcm\{{\{|C|: C\mbox{ is a cycle in }D(S)\}\cup\{p\}\}}$.
Then there exists a constant $c(\vH,S)$ such that whenever $\lambda_{\rm GRC}|(n-\ell+1)$,
\[
|\C(\vH,\vzero)\cap\pQ(n;S)|\ge c(\vH,S){n}^D+O(n^{D-1}),
\]
where $D=\sizeof{S} - \sizeof{V(S)}=\sum_{w=w_1}^{w_2}\binom{\ell}{w}(q^*)^w(q-q^*)^{\ell-w}-
\sum_{w=w_1-1}^{w_2}\binom{\ell-1}{w}(q^*)^w(q-q^*)^{\ell-1-w}$.  
\end{cor}

\begin{exa}\label{exa:intersect}
Let $S=\bbracket{2}^3$ and $d=2$.
Choose $p=13$ and 
\[\vH=\left(\begin{array}{cccc cccc}
1 & 2 & 3 & 5 & 8 & 10 & 11 & 12 \\
1 & 4 & 9 & 12 & 12 & 9 & 4 & 1
\end{array}\right).\]
Then $\C(\vH,\vzero)$ is an $(8,3)$-AECC containing $\vone$.
We have $\lambda_{\rm GRC}=\lcm\{{\{1,2,\ldots,8\}\cup\{13\}\}}=156$.
Using {\tt LattE}, we compute the lattice point enumerator of $\lambda_{\rm GRC}\Pgrcint(\vH,S)$ to be
$12168  t^{4} - 1248  t^{3} + 131 t^{2} - 16 t + 1$.
Hence, for $n=156t+2$, 
the number of codewords in 
$\C(\vH,\vzero)\cap\E(n;2,3)$ is given by $12168 t^{4} - 1248 t^{3} + 131 t^{2} - 16 t + 1$.
When $t=1$ or $n=158$, there exist a $(158,3;2,3)$-GRC of size at least $11036$.

We compare this result with the one provided by the construction using the systematic encoder described in Section \ref{sec:sys}
and in particular, Example \ref{exa:systematic}.
When $n=158$, we can systematically encode words in $\bbracket{39}^3$ into $\pQ(158;2,3)$.
Hence, we consider a $39$-ary $(3,3)$-AECC. Using Varshamov's construction with $p_1=5$ and 
$\vH_1=\left(\begin{array}{ccc}
1 & 2 & 3 \\
1 & 4 & 4
\end{array}\right)$, we obtain a $39$-ary $(3,3)$-AECC of size $2368$.
Applying the systematic encoder in Theorem \ref{thm:sys}, we construct a $(158,3;2,3)$-GRC of size $2368$.
\end{exa}

\begin{table*}[!t]
\caption{Computations of $c(\vH,S)$}
\label{tab:cqld}
When $S=\bbracket{2}^3$, we have $c(2,3)=1/288$. 
\begin{center}
\begin{tabular}{ccccccc}\hline
$d$ & $p$ & $D$ & $\lambda_{\rm GRC}$ & $c(\vH,S)$ & $c(2,3)/p^d$ \\ \hline
1 & 11 & 4 & 132 & 1/3168 & 1/3168 \\
2 & 13 & 4 & 156 & 1/48672 & 1/48672 \\
3 & 13 & 4 & 156 & 1/632736 & 1/632736 \\
4 & 17 & 4 & 204 & 1/24054048 & 1/24054048 \\
5 & 17 & 4 & 204 & 1/24054048 & 1/408918816 \\
6 & 17 & 4 & 204 & 1/24054048 & 1/6951619872 \\
\hline
\end{tabular}
\end{center}

When $S=\bbracket{2}^4$, we have $c(2,4)=283/9754214400$. 
\begin{center}
\begin{tabular}{cccccc}\hline
$d$ & $p$ & $D$ & $\lambda_{\rm GRC}$ & $c(\vH,S)$ & $c(2,4)/p^d$ \\ \hline
1 & 17 & 8 & 14280 & 283/165821644800 & 283/165821644800 \\
2 & 17 & 8 & 14280 & 283/2818967961600 & 283/2818967961600 \\
3 & 17 & 8 & 14280 & 283/47922455347200 & 283/47922455347200 \\\hline
\end{tabular}
\end{center}

When $S=S(2,5;1,[2,3])$, we have $c(2,5;1,[2,3])=1/5184000$.

\begin{center}
\begin{tabular}{cccccc}\hline
$d$ & $p$ & $D$ & $\lambda_{\rm GRC}$ & $c(\vH,S)$ & $c(2,5;1,[2,3])/p^d$ \\ \hline
1 & 23 & 6 & 2760 & 1/119232000 & 1/119232000 \\
2 & 29 & 6 & 3480 & 1/4359744000 & 1/4359744000 \\
3 & 29 & 6 & 3480 & 1/126432576000 & 1/126432576000 \\ \hline
\end{tabular}
\end{center}
 \end{table*}
 
Using {\tt LattE}, we determined $c(\vH,S)$ for moderate parameter values and summarize the results in Table \ref{tab:cqld}.

We conclude this section with a conjecture on the relation between $c(q,\ell)$ and $c(\vH,S)$.
\begin{conj}
Fix $q,\ell,d$. Choose $\vH$ and $p$ such that $\C(\vH,\vzero)$ is an $(N,d+1)$-AECC containing $\vone$.
Let $c(q,\ell)$ and $c(\vH,S)$ be the constants defined in Corollaries \ref{cor:cql} and \ref{cor:cHS}, respectively.
Then $c(\vH,S)\ge c(q,\ell)/p^d$.
\end{conj}

Roughly speaking, the conjecture states that asymptotically,
$|\C(\vH,\vzero)\cap \E(n;q,\ell)|$ is at least $|\Qb(n;q,\ell)|/p^d$.
In other words, for our particular choice of $\vH$ and $\vbeta$, 
we asymptotically achieve the code size guaranteed by the pigeonhole principle.

\section{Decoding of Profile Vectors}\label{sec:decoding}

Recall the DNA storage channel illustrated by Fig. \ref{fig:DNAstorage}. 
The channel takes as its input a word $\vx\in\Q(n;S)$ 
and outputs a vector $\vhx\in\ZZ^{|S|}$. 
Assuming no errors, the vector $\vhx$ corresponds to the profile vector $\vp(\vx;S)\in\pQ(n;S)$.
In this channel model and the code constructions in Section \ref{sec:lower}, 
we have implicitly assumed the existence of an efficient algorithm that 
decodes $\vp(\vx;S)$ back to the message $\vx$.
We explicitly describe this algorithm in what follows.

Let $\vu$ be a profile vector in $\pQ(n;S)$ so that $\vu=\vp(\vx;S)$ for some  $\vx\in\Q(n;S)$.
As with the proof of Lemma \ref{lem:euler}, we construct a multigraph on 
the node set $V(S)$ by adding $u_\vz$ arcs for each $\vz\in V(S)$.
We remove any isolated vertices and we have a connected Eulerian multidigraph.
We subsequently apply any linear-time algorithm like Hierholzer's algorithm \cite{Hierholzer:1873} 
to this multidigraph to obtain an Eulerian walk%
\footnote{Most descriptions of the Hierholzer's algorithm involve an arbitrary choice for the starting vertex and 
the subsequent vertices to visit. Hence, it is possible for the algorithm to yield different walks based on the same multigraph.
Nevertheless, we may fix an order for the vertices and have the algorithm always choose the `smallest' available vertex. Under these assumptions, $\textsc{Euler}(\vu)$ is always well defined.}
and let $\textsc{Euler}(\vu)$ denote the word of $\bbracket{Q}^n$ obtained from this Eulerian walk.
It remains to verify that $\textsc{Euler}(\vu)=\vx$.

As mentioned in Section \ref{sec:profile}, an element in $\Q(n;S)$ is an equivalence class $X\subset \bbracket{q}^n$, where $\vx,\vx'\in X$ implies that $\vp(\vx;S)=\vp(\vx';S)$.
Here, we fix the choice of {representative} for $X$. 
As hinted by the previous discussion, we let this representative be $\textsc{Euler}\left(\vp(\vy;S)\right)$ for some $\vy\in Y$ and observe that this definition is independent of the choice of $\vy$.
Then with this choice of representatives, the function $\textsc{Euler}$ indeed decodes a profile vector back to its representative codeword.

In summary, we identify the elements in $\Q(n;S)$ with the set of representatives
$\{\textsc{Euler}(\vu):\vu\in\pQ(n;S\}$.
Then for any $\vx\in\Q(n;S)$, the function $\textsc{Euler}$ decodes $\vp(\vx;S)$ to $\vx$ in linear-time.

An interesting feature of our coding scheme is that we avoided the assembly problem 
by designing our codewords to have distinct profile vectors and profiles at sufficiently large distance.
However, there are challenges in counting accurately the number of $\ell$-grams and 
determining the profile vector of  an arbitrary word using current sequencing technologies.
We examine next a number of practical methods for profile counting and decoding and 
address emerging issues via known coding solutions.
In our discussion, we assume that $S=\bbracket{q}^\ell$.

In particular, we look at an older technology -- sequencing by hybridization (SBH), proposed in~\cite{Pevzner.Lipshutz:1994} -- as a means of automated decoding. 
The idea behind SBH is to build an array of $\ell$-grams or {\em probes}; this array of probes is commonly referred to as a {\em sequencing chip}. 
A sample of single stranded DNA to be sequenced 
is fragmented, labelled with a radioactive or fluorescent material, and then presented to the chip. 
Each probe in the array hybridizes with its reverse complement,
provided the corresponding $\ell$-gram is present in the sample. 
Then an optical detector measures the intensity of hybridization of the labelled DNA and 
hence infers the number of $\ell$-grams present in the sample.

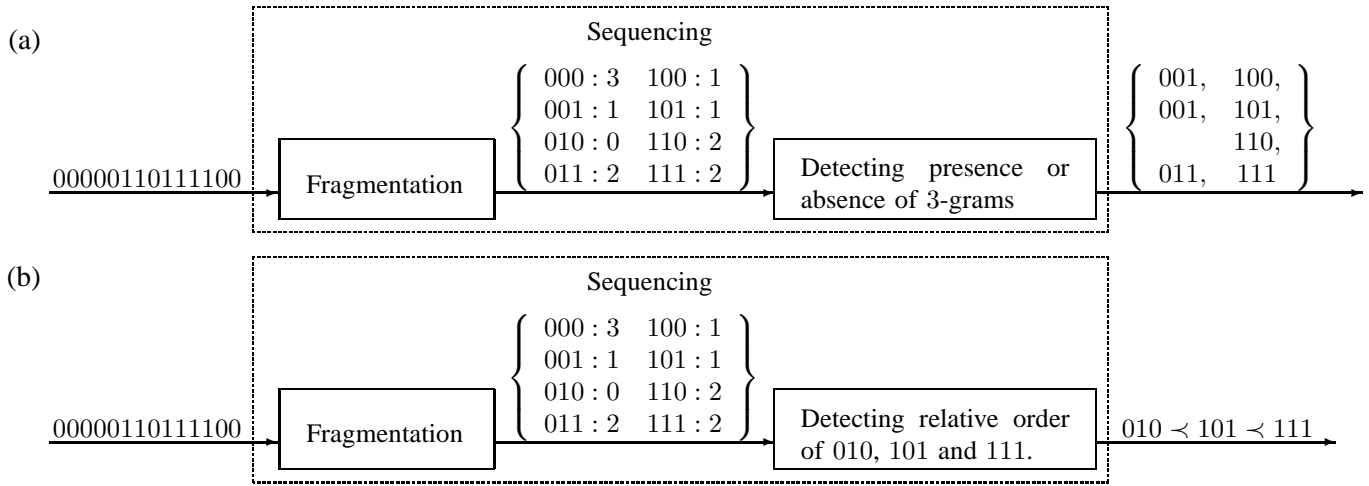
\begin{figure*}[t]
\begin{enumerate}[(a)]
\item \hfill
\begin{center}
\small
\begin{picture}(480,70)(0,-10)

\put(75,0){\framebox(80,30){}}
\put(85,10){\mbox{Fragmentation}}

\put(260,0){\framebox(120,30){}}
\put(270,10){\parbox[c][4em][c]{100px}{Detecting presence or absence of $3$-grams}}

\put(190,70){\txt{Sequencing}}
\put(65,-5){\dashbox(320,85){}}

\put(-11,10){\vector(1,0){85}}
\put(-10,12){\mbox{$0000 0110 1111 00$}}
\put(155,10){\vector(1,0){105}}
\put(160,32){\mbox{$\left\{\begin{array}{cc}
000:3 & 100:1\\
001:1 & 101:1\\
010:0 & 110:2\\
011:2 & 111:2\end{array}\right\}$}}

\put(380,10){\vector(1,0){100}}
\put(390,32){\mbox{$\left\{\begin{array}{cc}
001, & 100,\\
001, & 101,\\
 & 110,\\
011, & 111\end{array}\right\}$}}
\end{picture}
\end{center}

\item \hfill
\begin{center}
\small
\begin{picture}(480,70)(0,-5)

\put(75,0){\framebox(80,30){}}
\put(85,10){\mbox{Fragmentation}}

\put(260,0){\framebox(120,30){}}
\put(270,10){\parbox[c][4em][c]{100px}{Detecting relative order of $010$, $101$ and $111$.}}

\put(190,70){\txt{Sequencing}}
\put(65,-5){\dashbox(320,85){}}

\put(-11,10){\vector(1,0){85}}
\put(-10,12){\mbox{$0000 0110 1111 00$}}
\put(155,10){\vector(1,0){105}}
\put(160,32){\mbox{$\left\{\begin{array}{cc}
000:3 & 100:1\\
001:1 & 101:1\\
010:0 & 110:2\\
011:2 & 111:2\end{array}\right\}$}}

\put(380,10){\vector(1,0){90}}
\put(390,12){\mbox{$010\prec 101\prec 111$}}

\end{picture}
\end{center}
\end{enumerate}

\caption{Sequencing by hybridization. Instead of obtaining the exact count of the $\ell$-grams, we obtain auxiliary information on the count: (a) we obtain the set of $3$-grams present in $001 110 110 000 00$; 
(b) we obtain the relative order of the counts of $010$, $101$ and $111$.  }
\label{fig:SBH}
\end{figure*}

\subsection{Detecting presence or absence of $\ell$-grams}

In the initial studies of SBH, hybridization results only indicated the presence or absence of certain $\ell$-grams.
In our terminology, if $\vx$ is the codeword, 
the channel outputs a subset of $\bbracket{q}^\ell$ given by $\supp(\vp(\vx;q,\ell))$,
where $\supp(\vu)$ denotes the set of coordinates $\vz$ with $u_\vz\ge 1$ (see Fig. \ref{fig:SBH}(a)).
Then, we can define $\dg^*(\vx,\vy;q,\ell)\triangleq|\supp(\vp(\vx;q,\ell))\Delta\supp(\vp(\vy;q,\ell))|$ 
for any pair of $\vx,\vy\in\bbracket{q}^n$.

As before, $(\bbracket{q}^n,\dg^*)$ forms a pseudometric space and we convert this space into a metric space via an equivalence relation --
we say $\vx\stackrel{\ell^*}{\sim}\vy$ if and only if $\dg^*(\vx,\vy;q,\ell)=0$. 
Then, by defining $\Q^*(n;q,\ell)\triangleq\bbracket{q}^n/\stackrel{\ell^*}{\sim}$, we obtain a metric space.

Let $\C\subseteq \Q^*(n;q,\ell)$. If $d=\min\{\dg^*(\vx,\vy;\ell): \vx,\vy\in \C, \vx\ne\vy\}$, then
$\C$ is said to be $(n,d;q,\ell)$-$\ell^*$-gram reconstruction code ($*$-GRC).

We have the following proposition that is an analogue of Proposition~\ref{prop:code}.

\begin{prop}
Given an $(n,d;q,\ell)$-$*$-GRC, a set of $n-\ell+1-\floor{(d-1)/2}$ $\ell$-grams suffices to identify a codeword.
\end{prop}

\begin{proof}
Let $t=n-\ell+1-\floor{(d-1)/2}$.
Suppose otherwise that there exists a pair of distinct codewords $\vx$ and $\vy$ that 
contain a common set of $t$ $\ell$-grams.
Then 
\begin{align*}
\dg^*(\vx,\vy;\ell)&=|\supp(\vp(\vx;q,\ell))\Delta\supp(\vp(\vy;q,\ell))|\\
&\le (n-\ell+1-t)+(n-\ell+1-t)=2\floor{(d-1)/2}) \le d-1<d,
\end{align*}
resulting in a contradiction.
\end{proof}

Determining the maximum size of an $(n,d;q,\ell)$-$*$-GRC turns out to be related to certain well studied combinatorial problems.

{\bf Case $d=1$}. The maximum size of an $(n,1;q,\ell)$-$*$-GRC is given by $|\Q^*(n;q,\ell)|$. 
Equivalently, this count corresponds to the number of possible sets of $\ell$-grams that can be obtained from words of length $n$.
Observe that $|\Q^*(n;q,\ell)|\le 2^{q^\ell}$ and hence $|\Q^*(n;q,\ell)|$ cannot be a quasipolynomial in $n$ 
with degree at least one. Therefore, it appears that Ehrhart theory is not applicable in this context.
Nevertheless, preliminary investigations of this quantity for $q=2$ have been performed by Tan and Shallit \cite{tan2013sets}. In particular, Tan and Shallit proved the following proposition for $n<2\ell$. 

\begin{prop}[{\cite[Corollary 19]{tan2013sets}}]
For $\ell\le n <2\ell$, we have 
\[\Q(n,\ell)=2^n-\sum_{k=1}^{n-\ell+1}\frac{k-1}{k}\sum_{d|k}\mu\left(\frac kd\right)2^d,\]
where $\mu(\cdot)$ is the M\"obius function defined as
\[
\mu(n)=\begin{cases}
1, & \mbox{if $n$ is a square-free positive integer with an even number of prime factors};\\
1, & \mbox{if $n$ is a square-free positive integer with an odd number of prime factors};\\
0, & \mbox{otherwise}.
\end{cases}
\]
\end{prop}

{\bf Case $d=2(n-\ell+1)$}. For the other extreme, we see that the problem is related 
to edge-disjoint path packings and decompositions of graphs (see \cite{Heinrich:1993,Bryant.El-Zanati:2007}).
Formally, consider a graph $G$. A set $\C$ of paths in $G$ is said to be an {\em edge-disjoint path packing} of $G$
if each edge in $G$ appears in at most one path in $\C$.
An edge-disjoint path packing $\C$ of $G$ is an {\em edge-disjoint path decomposition} of $G$
if each edge in $G$ appears in exactly one path in $\C$.
Edge-disjoint cycle packings and decompositions are defined similarly.

Now, an $(n,2(n-\ell+1);q,\ell)$-$*$-GRC is equivalent to an edge-disjoint path packing of $D(q,\ell)$, where each path is of length $(n-\ell+1)$.
Furthermore, an edge-disjoint path decomposition of $D(q,\ell)$ into paths of length $n-\ell+1$ 
yields an optimal $(n,2(n-\ell+1);q,\ell)$-$*$-GRC of size $q^\ell/(n-\ell+1)$.

Since an edge-disjoint cycle decomposition is also an edge-disjoint path decomposition,
we examine next edge-disjoint cycle decomposition of de Bruijn graphs.
These combinatorial objects were studied by Cooper and Graham, who proved the following theorem.

\begin{thm}[{\cite[Proposition 2.3, Corollary 2.5]{cooper2004generalized}}]\label{cooper}\hfill
\begin{enumerate}
\item There exists an edge-disjoint cycle decomposition of  $D(q,\ell)$ into $q$ cycles of length $q^{\ell-1}$, for any $q$ and $\ell$. 
\item There exists an edge-disjoint cycle decomposition of  $D(r2^{k+1},3)$ into $8^k$ cycles of length $8r^3$, for any $k\ge 0$ and $r\ge 1$.
\end{enumerate}
\end{thm}

Therefore, Theorem \ref{cooper} demonstrates the existence of 
an optimal $(q^{\ell-1} +\ell- 1, 2q^{\ell-1};q,\ell)$-$*$-GRC of size $q$
and an optimal  $(8r^3+2,16r^3;r2^{k+1},3)$-$*$-GRC of size $8^k$ for any $k\ge 0$ and $r\ge 1$.

\subsection{Detecting the relative order of $\ell$-grams}\label{sec:rankmod}

As mentioned earlier, it is difficult to infer accurately the number of $\ell$-grams present from the hybridization results.  However, we may significantly more accurately determine whether the count of a certain $\ell$-gram is greater than the count of another.
In other words, we may view the sequencing channel outputs as {\em rankings} or {\em orderings} on the $q^\ell$ $\ell$-grams counts or a {\em permutation} of length $q^\ell$ reflecting the $\ell$-gram counts.

This suggests that we consider codewords whose profile vectors carry information about order.
More precisely, let ${\rm Perm}(N)$ denote the set of permutations over the set $\bbracket{N}$. 
We consider codewords whose profile vectors belong to ${\rm Perm}(N)$
and consider a metric on ${\rm Perm}(N)$ that relates to errors resulting from changes in order.
The Kendall metric was first proposed by Jiang \etal{} \cite{Jiang.etal:2009} 
in {rank modulation schemes} for nonvolatile flash memories and 
codes in this metric have been studied extensively since 
(see \cite{Barg.Mazumdar:2010} and the references therein). 
The Ulam metric was later proposed by Farnoud \etal{} for permutations \cite{Farnoud.etal:2013}
and multipermutations \cite{Farnoud.Milenkovic:2014}.

Unfortunately, due to the flow conservation equations \eqref{eq:flow}, 
the profile vector of a $q$-ary word is unlikely to have distinct entries and hence be a permutation.
Nevertheless, we appeal to the systematic encoder provided by Theorem \ref{thm:sys}.
We set $m=q^\ell-q^{\ell-1}-1$. Then, provided $n$ is sufficiently large,
there exists a set $I$ of $m$ coordinates that allow us to
extend any word $\vv$ in $\bbracket{m}^m$ to a profile vector in $\phi_{\rm sys}(\vv)\in\pQ(n;q,\ell)$.
In particular, since ${\rm Perm}(m)\subseteq\bbracket{m}^m$, \emph{any permutation} $\vv$ of length $m$ 
may be extended to a profile vector in $\phi_{\rm sys}(\vv)\in\pQ(n;q,\ell)$.

This implies that for the design of the sequencing chip,
we do not need to have $q^\ell$ probes for all possible $\ell$-grams.
Instead, we require only $m=q^\ell-q^{\ell-1}-1$ probes that correspond to the $\ell$-grams in $I$.
Hence, the sequencing channel outputs an ordering on this set of $m$ $\ell$-grams 
(see Fig. \ref{fig:SBH}(b)).

This setup allows us to integrate known rank modulation codes (in any metric) 
into our coding schemes for DNA storage. 
In particular, to encode information we perform the following procedure.
First, we encode a message is into a permutation using a rank modulation encoder.
Then the permutation is extended into a profile vector and then mapped by
{\sc Euler} to the profile vector of a $q$-ary codeword
(see Fig. \ref{fig:ranking} for an illustration).

\begin{exa}
Suppose that $S=\bbracket{2}^3$.
Hence, we set $m=3$ and
recall the systematic encoder $\phi_{\rm sys}$ described in Example \ref{exa:systematic}
that maps $\bbracket{3}^3$ into $\pQ(14;2,3)$.
Suppose that $\vv=(0,1,2)\in {\rm Perm}(3)$ belongs to some rank modulation code.
Then $\vu=\phi_{\rm sys}(\vv)=(3,1,{\bf 0},2,{\bf1},1,2,{\bf 2})$ belongs to $\pQ(14;2,3)$.
Finally, {\sc Euler} maps $\vu$ to a codeword $0000 0110 1111 00\in\bbracket{2}^{14}$.

Now, if we were to detect the relative order of the $3$-grams $010$, $101$ and $111$,
we obtain the permutation $(0,1,2)$ as desired (see also Fig. \ref{fig:SBH}(b)).
\end{exa}

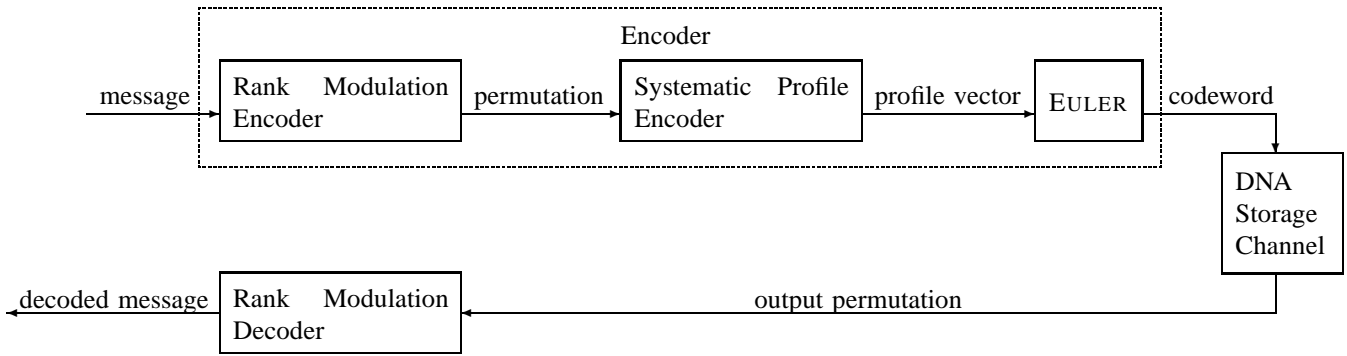
\begin{figure*}[t]

\begin{center}
\small
\begin{picture}(480,120)(-20,-50)

\put(50,30){\framebox(90,30){}}
\put(55,42){\parbox[c][4em][c]{80px}{Rank Modulation Encoder}}

\put(200,30){\framebox(90,30){}}
\put(205,42){\parbox[c][4em][c]{80px}{Systematic Profile Encoder}}

\put(355,30){\framebox(40,30){}}
\put(360,42){\parbox[c][4em][c]{30px}{\sc Euler}}

\put(200,70){\txt{Encoder}}
\put(42,20){\dashbox(360,60){}}

\put(425,-20){\framebox(45,45){}}
\put(430,0){\parbox[c][4em][c]{40px}{DNA Storage Channel}}

\put(0,40){\vector(1,0){50}}
\put(5,44){\mbox{message}}

\put(140,40){\vector(1,0){60}}
\put(145,44){\mbox{permutation}}

\put(290,40){\vector(1,0){65}}
\put(295,44){\mbox{profile vector}}

\put(395,40){\line(1,0){50}}
\put(405,44){\mbox{codeword}}
\put(445,40){\vector(0,-1){15}}

\put(445,-20){\line(0,-1){15}}
\put(445,-35){\vector(-1,0){305}}
\put(250,-33){\mbox{output permutation}}

\put(50,-50){\framebox(90,30){}}
\put(55,-38){\parbox[c][4em][c]{80px}{Rank Modulation Decoder}}

\put(50,-35){\vector(-1,0){80}}
\put(-25,-33){\mbox{decoded message}}

\end{picture}
\end{center}

\caption{Encoding messages for a DNA storage channel that outputs the relative order on the counts of particular $\ell$-grams.  }
\label{fig:ranking}
\end{figure*}

\appendices

\section{Eulerian Property of Certain Restricted De Bruijn Digraphs}
\label{app:restricteddb}

In this section, we provide a detailed proof of Proposition \ref{debruijn}. 
Specifically, for $q$, $\ell$, $1\le q^*\le q-1$ and $1\le w_1<w_2\le \ell$, we demonstrate that
the digraph $D(q,\ell;q^*,[w_1,w_2])$ is Eulerian. Our analysis follows that of Ruskey \etal{} \cite{Ruskey.etal:2012}.

Recall that the arc set of $D(q,\ell;q^*,[w_1,w_2])$ is given by $S=S(q,\ell;q^*,[w_1,w_2])$, 
while the node set is given by $V(S)=S(q,\ell-1;q^*,[w_1-1,w_2])$, which we denote by $V$ for short.
In addition, we introduce the following subsets of $\bbracket{q}$.
For a node $\vz$ in $V$, 
let ${\rm Pref}(\vz)$ be the set of symbols in $\bbracket{q}$ that 
when prepended to $\vz$ results in an arc in $S$. 
Similarly, let ${\rm Suff}(\vz)$ be the set of symbols in $\bbracket{q}$ that 
when appended to $\vz$ result in an arc in $S$.
Hence, $\{\sigma\vz:\sigma\in {\rm Pref}(\vz)\}$ and $\{\vz\sigma:\sigma\in {\rm Suff}(\vz)\}$
are the respective sets of incoming and outgoing arcs for the node $\vz$.

\begin{lem}\label{lem:balanced}
Every node of $D(q,\ell;q^*,[w_1,w_2])$ has the same number of incoming and outgoing arcs.
\end{lem}

\begin{proof}
Let $\vz$ belong to $V$. Observe that for all $s \in\bbracket{q}$, $s\,\vz\in S$ 
if and only if $\vz\,s \in S$. 
Hence, ${\rm Pref}(\vz)={\rm Suff}(\vz)$ and the lemma follows.
\end{proof}

It remains to show that $D(q,\ell;q^*,[w_1,w_2])$ is strongly connected.
We do it via the following sequence of lemmas.

\begin{lem}
Let $\vz,\vz'$ belong to $V$ and have the property that
they differ in exactly one coordinate. 
Then there exists a path from $\vz$ to $\vz'$.
\end{lem}

\begin{proof}
Observe the following characterization of ${\rm Pref}(\vz)={\rm Suff}(\vz)$:

\begin{equation*}
{\rm Pref}(\vz)={\rm Suff}(\vz)=
\begin{cases}
[q-q^*,q-1], & \mbox{if $\wt(\vz;q^*)=w_1-1$};\\
\bbracket{q^*}, & \mbox{if $\wt(\vz;q^*)=w_2$};\\
\bbracket{q}, & \mbox{otherwise}.
\end{cases}
\end{equation*}
Then ${\rm Suff}(\vz)\cap{\rm Pref}(\vz')$ is empty only if $\wt(\vz;q^*)=w_1-1$ and $\wt(\vz';q^*)=w_2$.
Therefore, $\vz$ and $\vz'$ differ in at least two coordinates, which contradicts the starting assumption.

Hence, ${\rm Suff}(\vz)\cap{\rm Pref}(\vz')$ is always nonempty.
To complete the proof, let $s \in {\rm Suff}(\vz)\cap{\rm Pref}(\vz')$. Then, the path corresponding to $\vz\,s\,\vz'$ is the desired path.
(Note that each $\ell$-gram appearing in $\vz\, s\, \vz'$ has weight equal to either $\wt(\vz\,s)$ or $\wt(s\,\vz')$; in particular,
each such $\ell$-gram lies in $S$.)
\end{proof}

Therefore, to construct a path between any two given nodes $\vz$ and $\vz'$, 
it suffices to demonstrate a sequence of nodes such that consecutive nodes differ in only one position.

\begin{lem}
  For any $\vz, \vz' \in V$, there is a sequence of nodes
  $\vz=\vz_0,\vz_1,\ldots,\vz_t=\vz'$ such that $\vz_{j}$ and
  $\vz_{j+1}$ differ in exactly one position for $j\in\bbracket{t}$.
\end{lem}
\begin{proof}
Let $\vz'=\sigma_1\sigma_2\cdots\sigma_{\ell-1}$. We construct the sequence of nodes inductively.
Suppose that for some $j$, $\vz_j=\sigma_1\sigma_2\cdots \sigma_i\tau_{i+1}\cdots\tau_{\ell-1}$, with 
$\tau_{i+1}\ne \sigma_{i+1}$.
Our objective is to construct a sequence of nodes terminating at $\vz_{j'}$, such that 
$\vz_{j'}=\sigma_1\sigma_2\cdots \sigma_i\sigma_{i+1}\tau'_{i+2}\cdots\tau'_{\ell-1}$ for some $\tau'_{i+1},\tau'_{i+2},\ldots,\tau'_{\ell-1}$. Hence, by repeating this procedure, we obtain the desired sequence of nodes that terminates at $\vz'$.

For the inductive step, if $\sigma_1\sigma_2\cdots \sigma_i\sigma_{i+1}\tau_{i+2}\cdots\tau_{\ell-1}\in V$, 
we establish the claim by setting 
$$\vz_{j+1}=\sigma_1\sigma_2\cdots \sigma_i\sigma_{i+1}\tau_{i+2}\cdots\tau_{\ell-1}.$$ 
In the general case, we have to consider the following scenarios:

\begin{enumerate}
\item When $\wt(\vz_j;q^*)=w_1-1$, $\tau_{i+1}\in[q-q^*,q-1]$ and $\sigma_{i+1}\notin[q-q^*,q-1]$, 
there exists some $\tau_k$ in $\vz_j$ that does not belongs to $[q-q^*,q-1]$.
Otherwise, $\wt(\sigma_1\cdots\sigma_i;q^*)=w_1-\ell+i$ and 
so $\wt(\sigma_1\cdots\sigma_{i+1};q^*)=w_1-\ell+i$. 
Then, $\wt(\vz';q^*)\le w_1-2$, contradicting the fact that $\vz'\in V$.
Therefore, we have the sequences: 
$\vz_{j}=\sigma_1\cdots \sigma_i\tau_{i+1}\tau_{i+2}\cdots\tau_k\cdots\tau_{\ell-1}$,
$\vz_{j+1}=\sigma_1\cdots \sigma_i\tau_{i+1}\tau_{i+2}\cdots 1 \cdots\tau_{\ell-1}$, and
$\vz_{j+2}=\sigma_1\cdots \sigma_i\sigma_{i+1}\tau_{i+2}\cdots 1\cdots\tau_{\ell-1}$.

\item When $\wt(\vz_j;q^*)=w_2$, $\tau_{i+1}\notin[q-q^*,q-1]$ and $\sigma_{i+1}\in[q-q^*,q-1]$, 
then there exists some $\tau_k$ in $\vz_j$ that belongs to $[q-q^*,q-1]$.
Otherwise, $\wt(\sigma_1\cdots\sigma_i;q^*)=w_2$ and 
so $\wt(\vz';q^*)\ge \wt(\sigma_1\cdots\sigma_{i+1};q^*)=w_2+1$, 
contradicting the fact that $\vz'\in V$.
Therefore, we have the sequences: 
$\vz_{j}=\sigma_1\cdots \sigma_i\tau_{i+1}\tau_{i+2}\cdots\tau_k\cdots\tau_{\ell-1}$,
$\vz_{j+1}=\sigma_1\cdots \sigma_i\tau_{i+1}\tau_{i+2}\cdots 0 \cdots\tau_{\ell-1}$, and
$\vz_{j+2}=\sigma_1\cdots \sigma_i\sigma_{i+1}\tau_{i+2}\cdots 0\cdots\tau_{\ell-1}$.
\end{enumerate}
\end{proof}
 
Consequently, $D(q,\ell;q^*,[w_1,w_2])$ is strongly connected. Together with Lemma \ref{lem:balanced},
this result establishes that $D(q,\ell;q^*,[w_1,w_2])$ is Eulerian.

\section{Proof of Corollary \ref{cor:aperiodic}}
\label{app:aperiodic}

We provide next a detailed proof of Corollary \ref{cor:aperiodic}. 
Specifically, we demonstrate Proposition \ref{prop:leading} from which the corollary follows directly.
For the case that $S=\bbracket{q}^\ell$, Jacquet \etal{} established a similar result 
by analyzing a sum of multinomial coefficients. This type of analysis appears to be to complex for a 
general choice of $S$.
 
\begin{prop}\label{prop:leading}
Suppose that $D(S)$ is strongly connected and that it contains loops.
Let $t=n-\ell+1$, $D=|S|-|V(S)|$ and let the lattice point enumerator of $\PP(S)$ 
be $L_{\PP(S)}(t)=c_{D}(t)t^{D}+O(t^{D-1})$. Then, $c_D(t)$ is constant.
\end{prop}

To prove this proposition, we use the following straightforward lemma. 

\begin{lem}\label{lem:monotone}
Suppose that $D(S)$ is strongly connected and that it contains loops.
For all $t$, we have $L_{\PP(S)}(t+1)\ge L_{\PP(S)}(t)$.
\end{lem}

\begin{proof}
  It suffices to show that there is an injection from $\F(n;S)$ to
  $\F(n+1;S)$.  Suppose that $\vu\in\F(n;S)$, so that
  $\vA(S)\vu=t\vb$.  Fix a loop in $D(S)$ and consider the vector
  $\vchi(\vz)$, where $\vz$ is the arc corresponding to the loop.
  Then, $\vA(S)\vchi(\vz)=\vb$ and $\vA(S)(\vu+\vchi(\vz))=(t+1)\vb$.
  So, the map $\vu\mapsto \vu+\vchi(\vz)$ is an injection from
  $\F(n;S)$ to $\F(n+1;S)$.
\end{proof}

\begin{proof}[Proof of Proposition \ref{prop:leading}]
Lemma \ref{lem:monotone} demonstrates that $L_{\PP(S)}$ is a monotonically increasing function.
Intuitively, this implies that the coefficient of its dominating term $c_{D}(t)$ cannot be periodic with period greater than $1$. We prove this claim formally in what follows.

Suppose that $c_D$ is not constant and that it has period $\tau$.
Hence, there exists $t_a\not\equiv t_b \bmod \tau$ such that $c_D(t_a)=a_D$, $c_D(t_b)=b_D$  
and $a_D<b_D$.
Furthermore, define $a_i=c_i(t_a)$ and $b_i=c_i(t_b)$ for $0\le i\le D-1$, and consider the polynomial $\sum_{i=0}^D b_it^i-a_i(t+\tau)^i$.
By construction, this polynomial has degree $D$ and a positive leading coefficient.
Hence, we can choose $t_1\equiv t_a\bmod \tau$ and $t_2\equiv t_b\bmod \tau$ so that
$t_1\le t_2\le t_1+\tau$ and $\sum_{i=0}^D b_it_2^i-a_i(t_1+\tau)^i> 0$.
Consequently,
\[L_{\PP(S)}(t_1+\tau)=\sum_{i=0}^D c_i(t_1+\tau)(t_1+\tau)^i=\sum_{i=0}^D a_i(t_1+\tau)^i
<\sum_{i=0}^D b_it_2^i=L_{\PP(S)}(t_2),\]
\noindent contradicting the monotonicity of $L_{\PP(S)}$.
\end{proof}

\section{Properties of the Polytope $\Pgrc(\vH,S)$}
\label{app:pgrc}

We derive properties of the polytope $\Pgrc(\vH,S)$ described in Section \ref{sec:intersect}. 
In particular, under the assumption that $D(S)$ is strongly connected and $\C(\vH,\vzero)\cap \nullplus\vB(D(S))$ is nonempty, we demonstrate the following:
\begin{enumerate}[(C1)]
\item The dimension of the polytope $\Pgrc(\vH,S)$ is $|S|-|V(S)|$;
\item The interior of the polytope is given by $\{\vu\in\RR^{|S|+d}: \vA(\vH,S)\vu=\vb, \vu>\vzero\}$;
\item The vertex set of the polytope is given by
\[\left\{\left(\frac{\vchi(C)}{|C|},\frac{\vH\vchi(C)}{p|C|}\right): C \mbox{ is a cycle in }D(S)\right\}.\]
\end{enumerate}

 Since $\C(\vH,\vzero)\cap \nullplus\vB(D(S))$ is nonempty, let $\vu_0$ belong to this intersection.
 Then $\vH\vu_0\equiv \vzero \bmod p$, that is, $\vH\vu_0=p\vbeta$ for some $\vbeta>0$.
 Let $\mu = \vone\vu_0$. 
 If we set $\vu=\frac{1}{\mu}(\vu_0,\vbeta)$, then $\vA(\vH,S)\vu=\vb$, with $\vu>0$.
 
Observe that the block structure of $\vA(\vH,S)$ implies that it has rank $|V(S)|+d$.
Hence, the nullity of $\vA(\vH,S)$ is $|S|-|V(S)|$. 
As before, let $\vu_1,\vu_2,\ldots$, $\vu_{|S|-|V(S)|}$ be linearly independent vectors
that span the null space of $\vA(\vH,S)$.
Since $\vu$ has strictly positive entries, we can find $\epsilon$ small enough so that 
$\vu+\epsilon\vu_i$ belongs to $\Pgrc(\vH,S)$ for all $ i\in [|S|-|V(S)|]$. 
Therefore, $\{\vu,\vu+\epsilon\vu_1,\vu+\epsilon\vu_2,\ldots, \vu+\epsilon\vu_{|S|-|V(S)|}\}$ 
is a set of $|S|-|V(S)|+1$ affinely independent points in $\Pgrc(\vH,S)$. This proves claim (C1).

For the interior of  $\Pgrc(\vH,S)$, first consider $\vu'>\vzero$ such that $\vA(\vH,S)\vu'=\vb$.
For any $\vu''\in\Pgrc(\vH,S)$, we have $\vA(\vH,S)\vu''=\vb$ 
and hence, $\vA(\vH,S)(\vu'-\vu'')=\vzero$.
Since $\vu'$ has strictly positive entries, we choose $\epsilon$ small enough so that 
$\vu'+\epsilon(\vu'-\vu'')\ge\vzero$.
Therefore, $\vu'+\epsilon(\vu'-\vu'')$ belongs to $\Pgrc(\vH,S)$ and 
$\vu'$ belongs to the interior of $\Pgrc(\vH,S)$. 

Conversely, let $\vu'\in\Pgrc(\vH,S)$ with $u'_j=0$ for some coordinate $j$.
Let $\vu$ be as defined earlier, where $\vu\in\Pgrc(\vH,S)$ with $\vu>\vzero$. 
Hence, for all $\epsilon>0$, the $j$th coordinate of $\vu'+\epsilon(\vu'-\vu)$ 
is given by $-\epsilon u_j$, which is always negative. 
In other words, $\vu'$ does not belong to interior of $\Pgrc(\vH,S)$.
This characterizes the interior as described in claim (C2).

For the vertex set, observe that $\left\{\left(\frac{\vchi(C)}{|C|},\frac{\vH\vchi(C)}{p|C|}\right): C \mbox{ is a cycle in }D(S)\right\}\subseteq \Pgrc(\vH,S)$. 

Let $\vv\in\Pgrc(\vH,S)$ and suppose that $\vv=(\vv_1,\vv_2)$ is a vertex.
Since $\vv\in\Pgrc(\vH,S)$, we have $\vv_2=\frac1p \vH\vv_1$
and $\vB(D(S))\vv_1=\vzero$.
Proceeding as in the proof of Lemma \ref{lem:vertex-ps}, 
we conclude that $\vv_1=\vchi(C)/|C|$, for some cycle in $D(S)$ and 
hence, $\vv=\left(\frac{\vchi(C)}{|C|},\frac{\vH\vchi(C)}{p|C|}\right)$.

Conversely, we show that  for any cycle $C$ in $D(S)$,
$\left(\frac{\vchi(C)}{|C|},\frac{\vH\vchi(C)}{p|C|}\right)$
cannot be expressed as a convex combination of other points in $\Pgrc(\vH,S)$.
 Suppose otherwise. Then we consider the first $|S|$ coordinates and 
 we proceed as in the proof of Lemma \ref{lem:vertex-ps} to yield a contradiction.
This completes the proof of claim (C3).

%

\bibliographystyle{IEEEtran}
\bibliography{mybibliography.bib}
\end{document}